\documentclass[sn-mathphys-num]{sn-jnl}% 
\usepackage{siunitx}
\usepackage{graphicx}%
\usepackage{multirow}%
\usepackage{amsmath,amssymb,amsfonts}%
\usepackage{mathrsfs}%
\usepackage[title]{appendix}%
\usepackage{xcolor}%
\usepackage{textcomp}%
\usepackage{manyfoot}%
\usepackage{booktabs}%
\usepackage{algorithm}%
\usepackage{algorithmicx}%
\usepackage{algpseudocode}%
\usepackage{listings}%
\usepackage{natbib}
\usepackage{hyperref} % Must be loaded last to avoid issues
\usepackage{mathtools}
\usepackage{geometry}
\usepackage{url}
\usepackage{amsthm}
\usepackage{soul}
\newcommand{\comb}[1]{\textcolor{black}{#1}}

\newtheorem{theorem}{Theorem}% 
\newtheorem{corollary}{Corollary} 
\newtheorem{lemma}{Lemma}

\usepackage{enumitem}

\makeatletter
\newcommand{\algmargin}{\the\ALG@thistlm}
\algnewcommand{\parState}[1]{\State%
  \parbox[t]{\dimexpr\linewidth-\algmargin}{\strut #1\strut}}

\begin{document}

\title{Counting Tree-Like Multigraphs with a Given Number of Vertices and Multiple Edges}


\author[1]{\fnm{Muhammad} \sur{Ilyas}}\email{milyas@math.qau.edu.pk}
\author[1]{\fnm{Seemab} \sur{Hayat}}\email{shayat@qau.edu.pk}
\author*[1]{\fnm{Naveed Ahmed} \sur{Azam}}\email{azam@amp.i.kyoto-u.ac.jp}

\affil[1]{\orgdiv{Department of Mathematics}, \orgname{Quaid-i-Azam University}, \city{Islamabad}, \country{Pakistan}}


\abstract{The enumeration of chemical graphs is an important topic in cheminformatics and bioinformatics, particularly in the discovery of novel drugs. These graphs are typically either tree-like multigraphs or composed of tree-like multigraphs connected to a core structure. In both cases, the tree-like components play a significant role in determining the properties and activities of chemical compounds.
This paper introduces a method based on dynamic programming to efficiently count tree-like multigraphs with a given number $n$ of vertices and $\Delta$ multiple edges. 
The idea of our method is to consider multigraphs as rooted multigraphs by selecting their unicentroid or bicentroid as the root, and define their canonical representation based on maximal subgraphs rooted at the children of the root. 
This representation guarantees that our proposed method will not repeat a multigraph in the counting process. 
Finally, recursive relations are derived based on the number of vertices and multiple edges in the maximal subgraphs rooted at the children of roots. 
These relations lead to an algorithm with a time complexity of $\mathcal{O}(n^2(n + \Delta (n + \Delta^2 \cdot \min\{n, \Delta\})))$ and a space complexity of $\mathcal{O}(n^2(\Delta^3+1))$.
Experimental results show that the proposed algorithm efficiently counts the desired multigraphs with up to 170 vertices and 50 multiple edges in approximately 930 seconds, confirming its effectiveness and potential as a valuable tool for exploring the chemical graph space in novel drug discovery.}

%%================================%%
%% Sample for structured abstract %%
%%================================%%


\maketitle
\section{Introduction}
The counting and generation of chemical compounds have a rich history and play a vital rule in applications such as novel drug design~\cite{blum2009970, lim2020scaffold} and structure elucidation~\cite{meringer2013small}. One way to tackle the challenge of counting and generating chemical compounds is to consider it as the problem of enumerating graphs under specific constraints. Enumeration with constraints has been extensively explored for various purposes, including the virtual exploration of the chemical universe~\cite{fink2007virtual, mauser2007chemical} and the reconstruction of molecular structure from their signatures~\cite{faulon2003signature, hall1993design}. 

Numerous techniques have been developed to enumerate chemical compounds~\cite{benecke1995molgen+, peironcely2012omg}. These techniques can be categorized into two classes. The first class contain enumeration methods that focuses on general graph structure~\cite{benecke1995molgen+, peironcely2012omg}. The second class consists of techniques specifically focused on enumerating certain restricted chemical compounds~\cite{akutsu2013comparison}. A contribution to this area was made by 
 Cayley~\cite{cayley1875analytical}, who laid the foundation for counting alkane isomers. Since then, a variety of various approaches have emerged. 

Chemical graphs can be classified into two categories based on the presence of rings or cycles. A graph without any rings is known as a tree-like multigraph. 
Two representative examples, such as Aspartic Acid~{\tt C$_{4}$H$_{7}$NO$_{4}$} which lacks a ring structure and resembles a tree-like multigraphs properties, is shown in Fig.~\ref{fig:peptide}(a).  and Homovanillic Acid~{\tt C$_{9}$H$_{10}$O$_{4}$} which contains a ring structure shown in Fig.~\ref{fig:peptide}(b). Notably, the chemical properties and activities of the tree-like structure  are often influenced significantly by their attached components. 
For example, Zhu et al.~\cite{zhu2022molecular} developed a model to predict the chemical properties of graphs by analyzing the tree-like components connected to the cores of the underlying structures. This shows the role of substructures in determining the behaviour  of chemical compounds. 
\begin{figure}[h!]
   \centering  
   \includegraphics[width=0.6\textwidth]{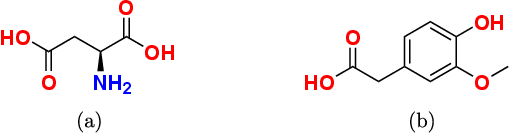}
  \caption{Chemical structures of compounds from PubChem database: (a) Aspartic Acid~{\tt C$_{4}$H$_{7}$NO$_{4}$}, a tree-like multigraph, PubChem CID: 5960; (b)~Homovanillic Acid~{\tt C$_{9}$H$_{10}$O$_{4}$}, which includes a ring structure, Pubchem CID: 1738.}
\label{fig:peptide}
\end{figure}

Two widely used techniques for enumerating chemical graphs are Pólya's enumeration theorem~\cite{polya2012combinatorial, polya1937kombinatorische} and branch-and-bound algorithms.
The key idea behind Pólya’s theorem is to count the number of distinct graphs by constructing a generating function based on cycle index of the symmetry group associated with the graphs. This method ensures that isomorphic graphs are counted only once, providing an efficient way for enumerating chemically relevant structures.
The computation of Pólya’s theorem heavily depends on the symmetry group of the underlying structure, which becomes challenging to handle for complex graphs. Calculating the cycle index for large or intricate symmetry groups can be computationally demanding.

Branch-and-bound algorithms have been extensively studies for enumeration of chemical graphs, particularly those tree-like structures. Akutsu and Fukagawa~\cite{akutsu2005inferring} showed that  problem of enumerating tree-like chemical graphs is NP-hard and developed a branch-bound algorithm to address this challenge~\cite{akutsu2007inferring}. Building on this work, Fujiwara et al.~\cite{fujiwara2008enumerating} proposed an improved algorithm  based on the two tree enumeration method by
Nakano and Uno~\cite{nakano2003efficient, nakano2005generating}, demonstrating comparable performance to the fatest existing algorithms~\cite{aringhieri2003chemical}. Subsequent studies introduced further refinements. Ishida et al.~\cite{ishida2008improved} presented two branch-and-bound algorithms for enumerating tree-like graphs using path frequencies, while Shimizu et al.~\cite{shimizu2011enumerating} and Suzuki et al.~\cite{suzuki20122} developed approaches for enumerating all tree-like (acyclic) chemical graphs based on feature vectors. Despite these advancements, branch-and-bound methods often suffer from high computational  costs, that compute the total number of solutions only after enumerating all relevant graphs. Furthermore, that frequently generate unnecessary intermediate structures, making them inefficient for large graphs and complex chemical graphs. While branch and bound are effective for small and medium sized tree-like compounds, their limitations highlight the need for more efficient approaches.

Dynamic programming (DP) is an efficient and widely used technique for discrete optimization problems~\cite{bellman1957dynamic}. Unlike branch-and-bound, which often generate unnecessary intermediate structure and take high computational costs for large graphs, Dp addresses these challenges  by breaking problems into smaller subproblems and establishing recursive relations between subproblems. This make DP particularly effective for problems involving graph structure. Akutsu et al.~\cite{akutsu2012inferring} utilized DP algorithms for tree-like graphs with a maximum degree bounded by a constant, providing a significant improvement in computational efficiency. Dp has also been applied to surface-embedded graphs with bounded branchwidth employing  surface cut decomposition techniques to reduce complexity to single-exponential dependence on branchwidth, $2^{\mathcal{O}(k)}.n$~\cite{rue2014dynamic}. This advancement allows DP to handle a broader range of graph structures efficiently.
In the context of chemical graph enumeration, DP has been studied in detail. For instance, it has been used to enumerate graphs efficiently~\cite{imada2011efficient, masui2009enumeration}.  
Recently, Azam et al.~\cite{azam2020efficient} extended this approach to count all tree-like graphs with a given number of vertices and self-loops. This shows that DP is efficient, especially in solving problems, where branch-and-bound, face difficulties.

In this paper, we focus on enumerating tree-like multigraphs with a given number of $n$ vertices, $\Delta$ multiple edges  and no self-loops, using DP. 
For this purpose, we treat each graph as a rooted graph by designating either a unicentroid or a bicentroid as the root. Recursive relations are derived based on the maximal subgraphs attached to the roots of the underlying graphs, which form the foundation of DP algorithm. To avoid redundant solutions, we introduce a canonical ordering for the maximal subgraphs. We analyze the theoretical time and space complexities of the proposed algorithm, which are 
$\mathcal{O}(n^2(n + \Delta (n + \Delta^2 \cdot \min\{n, \Delta\})))$ and $\mathcal{O}(n^2(\Delta^3+1))$, respectively. 
The proposed algorithm has been implemented to count the desired multigraphs for various values of $n$ and $\Delta$, demonstrating its efficiency in practical applications.
%\begin{figure}
%   \centering  
%   \includegraphics [height=3cm]{peptide_eps.eps}
%  \caption{(\textbf{a})~Chemical structure of Peptide~{\tt C$_{61}$H$_{114}$N$_{28}$O$_{16}$} with rings obtained from the PubChem database;
%(\textbf{b})~The chemical structure of Peptide~{\tt C$_{61}$H$_{114}$N$_{28}$O$_{16}$} is formed by joining the sub-molecules that resemble tree-like multigraphs with ring. ?? refine????(Add C9H10O4)}
%\label{fig:peptide_ring}
%\end{figure}

\noindent
The structure of the paper is as follows: Preliminaries are discussed in Section~\ref{preliminaries}. A recursive relation for enumerating tree-like multigraphs for a given number of vertices and multiple edges, which serves as a foundation for designing a DP-based algorithm for the counting are discussed in Section~\ref{counting}. Experimental results of tree-like multi graphs are discussed in Section~\ref{exp}. Conclusion with future directions is provided in Section~\ref{conclusion}.
%%%%%%%%%%%%%%%%%%%%%%%%%%%%%%%%%%%%%%%%%%
\section{Preliminaries}\label{preliminaries}
A \textit{graph} is an ordered pair $G$ = ($V$, $E$), where $V$ is a set of vertices and  $E$ is a set of edges. 
For a graph $G$, we denote by $V(G)$ and $E(G)$ the sets of vertices and edges, respectively. 
%A graph is visualized by drawing a circle for each vertex and joining two of these circles with a line if the corresponding two vertices form an edge. 
An edge between vertices $u$ and $v$ is denoted by $uv$. 
If two vertices $u$ and $v$ in a graph are the endpoints of an edge $uv$, they are said to be \textit{adjacent} vertices. 
An edge $uv$ is said to be incident to the vertices $u$ and $v$. 
An edge that starts and end at the same vertex, i.e., $uu$ is called a {\textit self-loop}. 
The {\em multiplicity} of an edge $uv$ is defined \comb{as the number of multiple edges between $u$ and $v$, and is denoted by  ${\rm e}(uv)$. In Fig.~\ref{fig:multigraph}, the multiplicity of the edge $v_1v_2$ is 0, as there are no multiple edges between $v_1$ and $v_2$. Similarly, the multiplicity of the edge $v_2v_3$ is 1, as there is one additional edge between $v_2$ and $v_3$.}
A graph with no self-loops and multiple edges is called {\textit{simple graph}. 
A graph with multiple edges or having self-loops is called a
\textit{multigraph}.  
Henceforth, we only consider graphs with multiple edges and no self-loops.

Let $M$ be a multigraph.  \comb{An illustration of a multigraph $M$ is given in Fig.}~\ref{fig:multigraph}.
The set of vertices incident to a vertex $v$ in $M$, excluding $v$ itself, is called the set of {\em neighbors} of $v$ and is denoted by $N(v)$. The {\it degree} of a vertex $v$ in $M$ is defined to be the size of $N(v)$ and is denoted by deg$_M(v)$. For instance, as shown in Fig.~\ref{fig:multigraph}, \comb{the neighbors of $v_2$ are $v_1$, $v_3$ and $v_4$, so deg$_M(v_2)=3$.} A \textit{subgraph} of $
M$ is a graph $H = (V', E')$ such that $V' \subseteq V$ and $E' \subseteq E$.
A \textit{path} $P(u, v)$ between vertices $u$ and $v$ is a sequence of vertices $w_1, w_2, \ldots, w_n$ where
$w_1=u$, $w_n = v$, and $w_iw_{i+1}$ is an edge for $i = 1, \ldots, n-1$. 
A path is called \textit{simple path} from $u$ to $v$ if each vertex in the path appears only once. 
A graph is said to be {\em connected} if  there exists a path between any two distinct vertices. 
A \textit{connected
component} of $M$ is a maximal connected subgraph of $M$.
Let $X \subseteq V(M)$. 
The induced subgraph $M[X]$ is the subgraph of $M$ with vertex set $X$ and edge set consisting of all edges in $E(M)$ that have  both endpoints in $X$.  A {\em cycle} is a path that starts and ends at the same vertex, contains at least three vertices, and has no repeated edges and vertices except for the starting and ending vertex.
\begin{figure}[h!]
    \centering
    \includegraphics[width=0.4\textwidth]{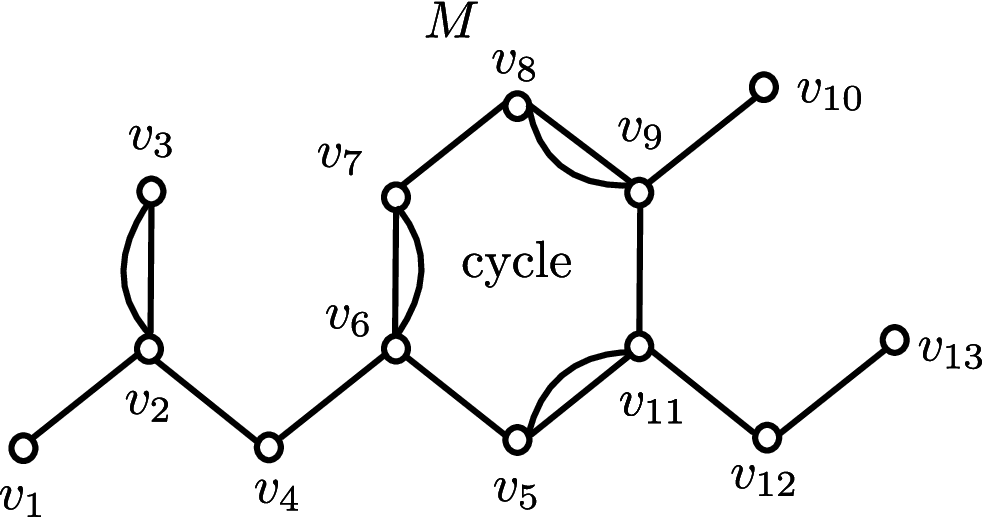}
    \caption{\comb{Illustration of a multigraph which is a 2D graphical structure of Homovanillic Acid shown in Fig.~\ref{fig:peptide}(b).}}
    \label{fig:multigraph}
\end{figure}
%That is, for any two vertices $u, v \in V_1$, $u$ and $v$ are adjacent in induced subgraph if and only if they are adjacent in $G$.
 
A \textit{tree-like multigraph} is a  connected multigraph with no cycles. 
A tree-like multigraph with a designated vertex is called a {\it rooted tree-like multigraph}. 
Let $M$ be a rooted tree-like multigraph. 
We denote by  $r_{M}$ the root of $M$. 
An {\it ancestor} of a vertex $v \in V(M) \setminus \{ r_M\}$ is a vertex, other than $v$, that lies on the path $P(v, r_M)$ from $v$ to the root $r_M$. Vertices for which $v$ serves as an ancestor are called the {\it descendants} of $v$. 
The {\it parent} $p(v)$ of a vertex $v \in V(M) \setminus \{ r_M\}$ is defined as the ancestor $u$ of $v$ such that $uv \in E(M)$. 
The vertex $v$ is referred to as a {\it child} of $p(v)$. 
Vertices in a rooted tree-like multigraph that share the same parent are called {\it siblings}. 
For a vertex $v \in V(M)$, let $M_v$ denote the descendant subgraph of $M_r$ rooted at $v$, which is the subgraph induced by $v$ and all its descendants.
Note that the descendant subgraph $M_v$ is a maximal rooted tree-like multigraph of $M$, rooted at $v,$ and induced by $v$ and its descendants. 
We denote by ${\rm v}(M_{v})$ the number of vertices and by ${\rm e}(M_{v})$, the number of \comb{ (multiple edges) in the descendant subgraph $M_v$. An illustration of a rooted tree-like multigraph is given in} Fig.~\ref{fig:multitree}.
\begin{figure}[h!]
    \centering
    \includegraphics[width=0.7\textwidth]{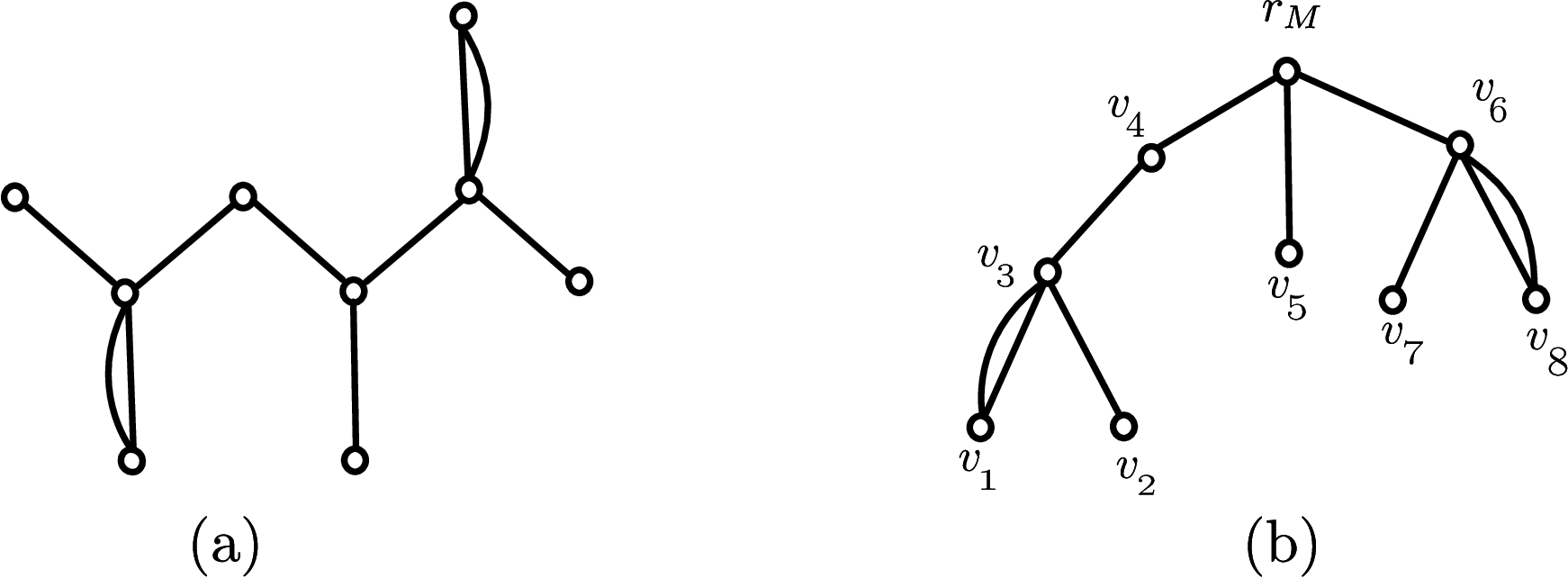}
    \caption{\comb{Rooted tree-like multigraph: (a) 2D graphical structure of Aspartic Acid shown in Fig.~\ref{fig:peptide}(a); (b) The multigraph with a specified root $r$. For vertex $v_3$, $v_4$ is its parent, $v_1$ and $v_2$ are its children. Vertex $v_3$ has no sibling, while $v_4$ and $r$ are its ancestors, and $v_1$ and $v_2$ are its descendants.}}
    \label{fig:multitree}
\end{figure}

By Jordan~\cite{jordan1869assemblages}, there exist either a  unique vertex or a unique edge in a tree with $n$ vertices, whose removal results in connected components with at most $\lfloor {(n-1)}/2 \rfloor$ or exactly $n/2$ vertices, respectively. 
This type of vertex is referred to as a {\it unicentroid}, \comb{ illustrated in Fig.~\ref{fig:centroid_article}(a) while such an edge is called the {\it bicentroid}, as shown in Fig.}~\ref{fig:centroid_article}(b)}. Collectively, they are known as the {\it centroid}. 
Note that a bicentroid only exists for an even number of vertices.
\begin{figure}[h!]
    \centering
    \includegraphics[width=0.85\textwidth]{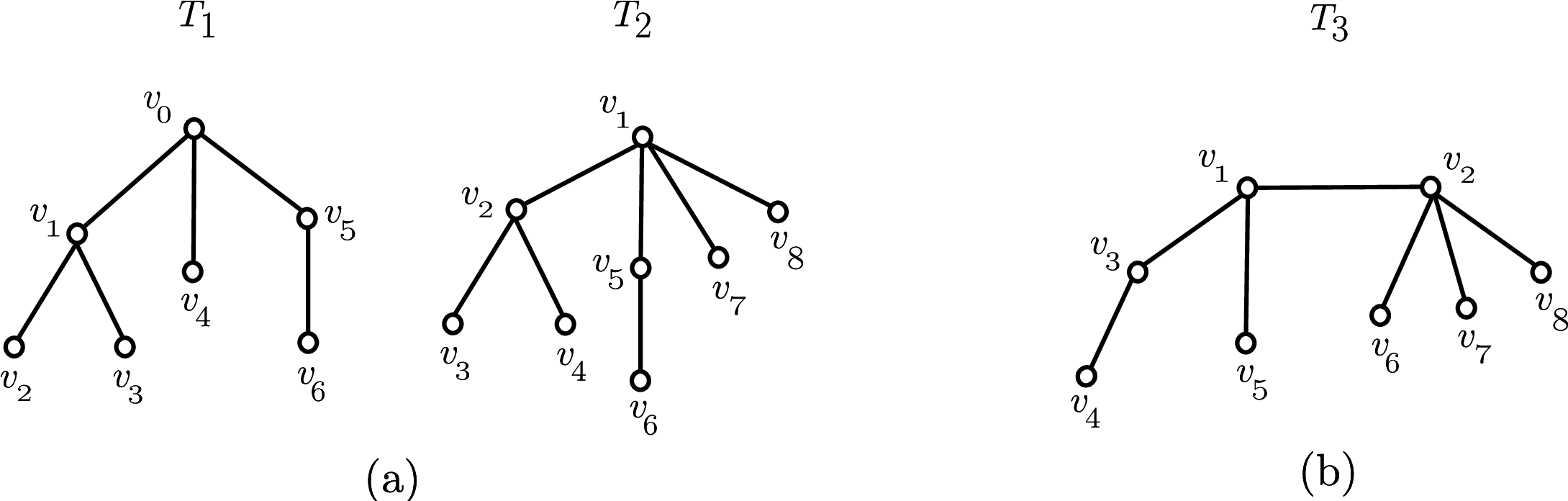}
    \caption{\comb{(a) Unicentroid $v_0$ in $T_1$ and $v_1$ in  $T_2$; (b) Bicentroid edge formed by vertex $v_1$ and $v_2$ in $T_3$. }}
    \label{fig:centroid_article}
\end{figure} 

Two rooted tree-like multigraphs, $M$ and $N$ with roots $r_M$ and $r_N$, resp., are said to be {\it isomorphic} if there exists a bijection $\varphi:~V(M) \to V(N)$ between their vertex sets such that $uv\in E(M)$ if and only if $\varphi(u)\varphi(v) \in E(N)$, 
$\varphi(r_M) = r_N$ and ${\rm e}(uv) = {\rm e}(\varphi(uv))$ for any two vertices $u, v\in V(M)$, where ${\rm e}(uv)$ denotes the number of edges between $u$ and $v$. An example of two isomorphic graphs are shown in Fig.~\ref{fig:isomorphictree}.
\begin{figure}[h!]
    \centering
    \includegraphics[width=0.7\textwidth]{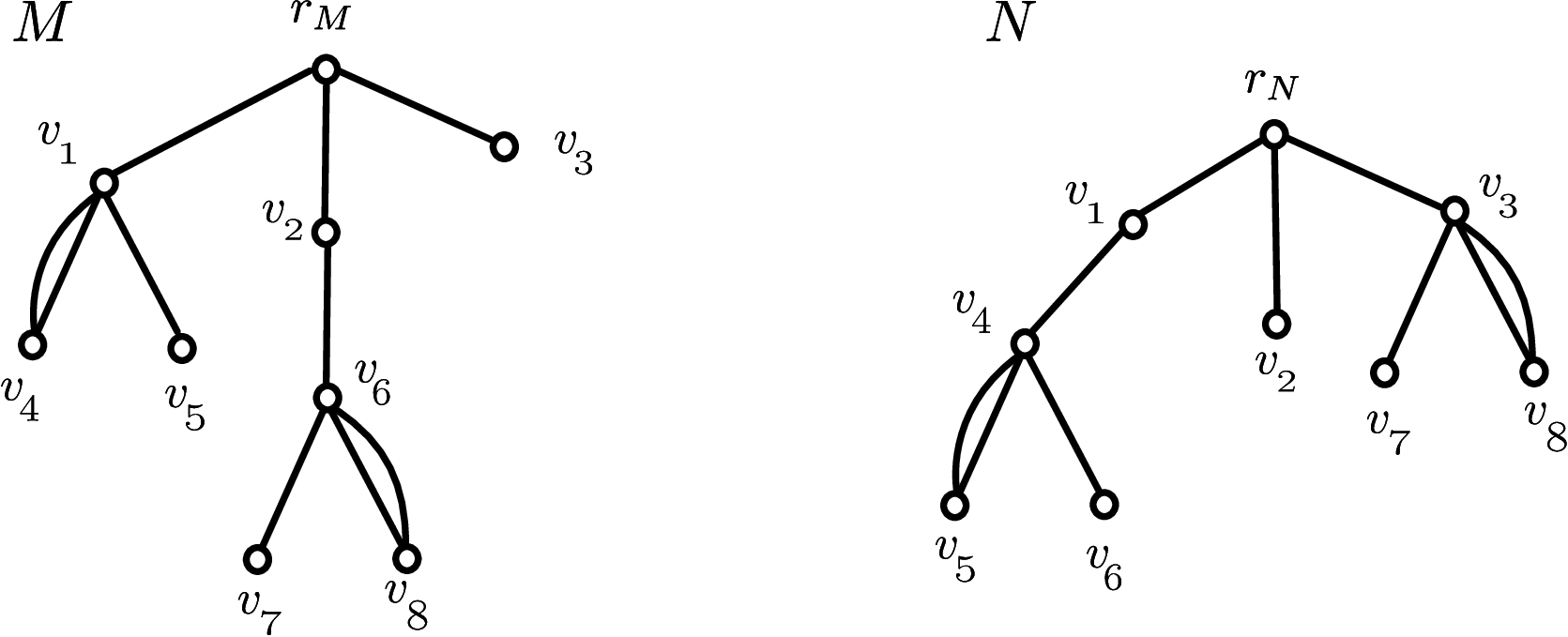}
    \caption{\comb{Illustration of rooted tree-like multigraphs  $M$ and $N$, which are isomorphic due to the mapping $\varphi$ such that $\varphi(r_M)=r_N$, $\varphi(v_1)=v_3$, $\varphi(v_2)=v_1$, $\varphi(v_3)=v_2$, $\varphi(v_4)=v_7$, $\varphi(v_5)=v_8$, $\varphi(v_6)=v_4$, $\varphi(v_7)=v_6$ and $\varphi(v_8)=v_5$,.  }}
    \label{fig:isomorphictree}
\end{figure}

Let $n \geq 0$ and $\Delta \geq 0$ be any two integers. We denote by 
$\mathcal{M}(n, \Delta)$ the set of all mutually non-isomorphic rooted tree-like multigraphs with $n$ vertices and $\Delta$ multiple edges. The size of $\mathcal{M}(n, \Delta)$ is denoted by $m(n, \Delta)$. 
%DP is a strategy used for designing algorithm when it recursively breaks down
%the problem into sub-problems satisfying some recursive relations, and computes the solution of the original problem by using the solutions
%of the sub-problems. To find the overall solution of the problem from
%the solutions of the sub-problems, recurrence relations are used. This is the primary feature of DP.
\section{Proposed Method}\label{counting}
We focus on counting the number $m(n, \Delta)$ of mutually non-isomorphic tree-like multigraphs with a prescribed number of vertices, multiple edges and no self-loops. For $n = 0, 1$ and any $\Delta \geq 1$, 
$m(n, \Delta) = 0$, and $m(2, \Delta) = 1$. \comb{Fig.~\ref{fig:n_graphs}(a) illustrates the case $m(2, \Delta) = 1$. Members of the family $\mathcal{M}(3, 2)$ corresponding to $n=3$ and $\Delta=2$, are shown in Fig.~\ref{fig:n_graphs}(b).}
We develop a DP algorithm to compute $m(n, \Delta)$, which counts the number of 
non-isomorphic tree-like multigraphs.
\begin{figure}[h!]
   \centering  
   \includegraphics[width=0.55\textwidth]{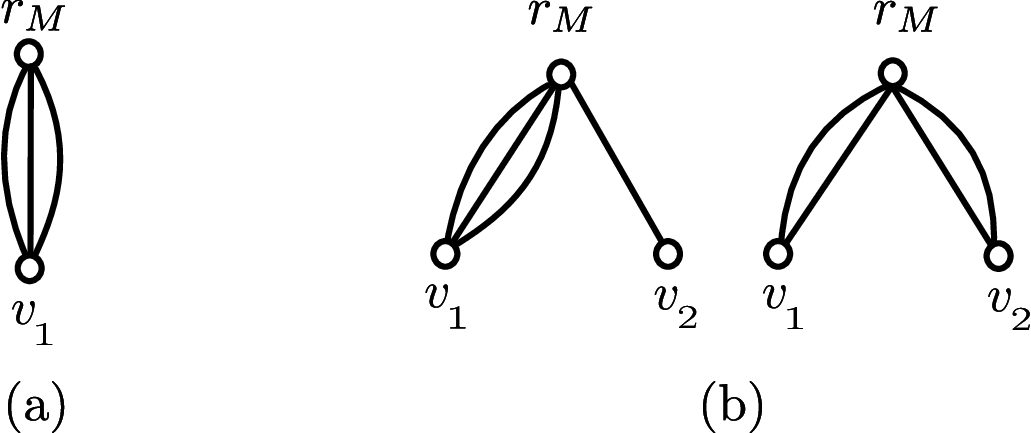}
  \caption{\comb{Illustration of the unique graph corresponding to $m(2, \Delta) = 1$ for $n=2$ and $\Delta\geq1$; (b) Examples of graphs of the family $\mathcal{M}(3, 2)$.} }
\label{fig:n_graphs}
\end{figure}
%%%

\noindent
\subsection{Canonical representation}  
For an enumeration method, it is essential to avoid generation of isomorphic graphs. 
For this purpose, we define a canonical representation for the graphs in 
$\mathcal{M}(n, \Delta)$. This representation utilizes information about the number of vertices and multiple edges in the descendant subgraphs, as well as  the number of multiple edges between the root and its children in the corrosponding underlying multigraphs. 
Formally, the {\it ordered representation} of a multigraph $M \in \mathcal{M}(n, \Delta)$ is defined as an ordered multigraph in which the descendant subgraphs are arranged from left-to-right according to the following criteria. \comb{ Consider a vertex of a multigraph and the descendant subgraphs of its children.  
(a) First order the descendant subgraphs in the non-increasing order of the number of vertices. 
(b) More than one subgraphs with the same number of vertices are further ordered in the non-increasing order of the number of multiple edges. 
(c) Lastly the subgraphs with the same number of multiple edges are arranged in the non-increasing order of the number of multiple edges between the root of the multigraph and the roots of descendant subgraphs.}
We define the {\it canonical representation} of a multigraph $M \in \mathcal{M}(n, \Delta)$ to be the ordered representation that satisfies (a)-(c) for each vertex.  
Henceforth, we assume that all graphs  in $\mathcal{M}(n, \Delta)$ are represented in their canonical form. 
For $n = 10 $ and $\Delta = 7 $, the canonical representation of a graph $M$ is shown in 
Fig.~\ref{fig:Tree_Like_Multigraphs}(a), while the non-canonical representation of $M$
is shown in Fig.~\ref{fig:Tree_Like_Multigraphs}(b). 
%%%%%%%%%%%%%%%%%%%%%%%%%%%%%%5
\begin{figure}[h!]
   \centering  
   \includegraphics[width=0.6\textwidth]{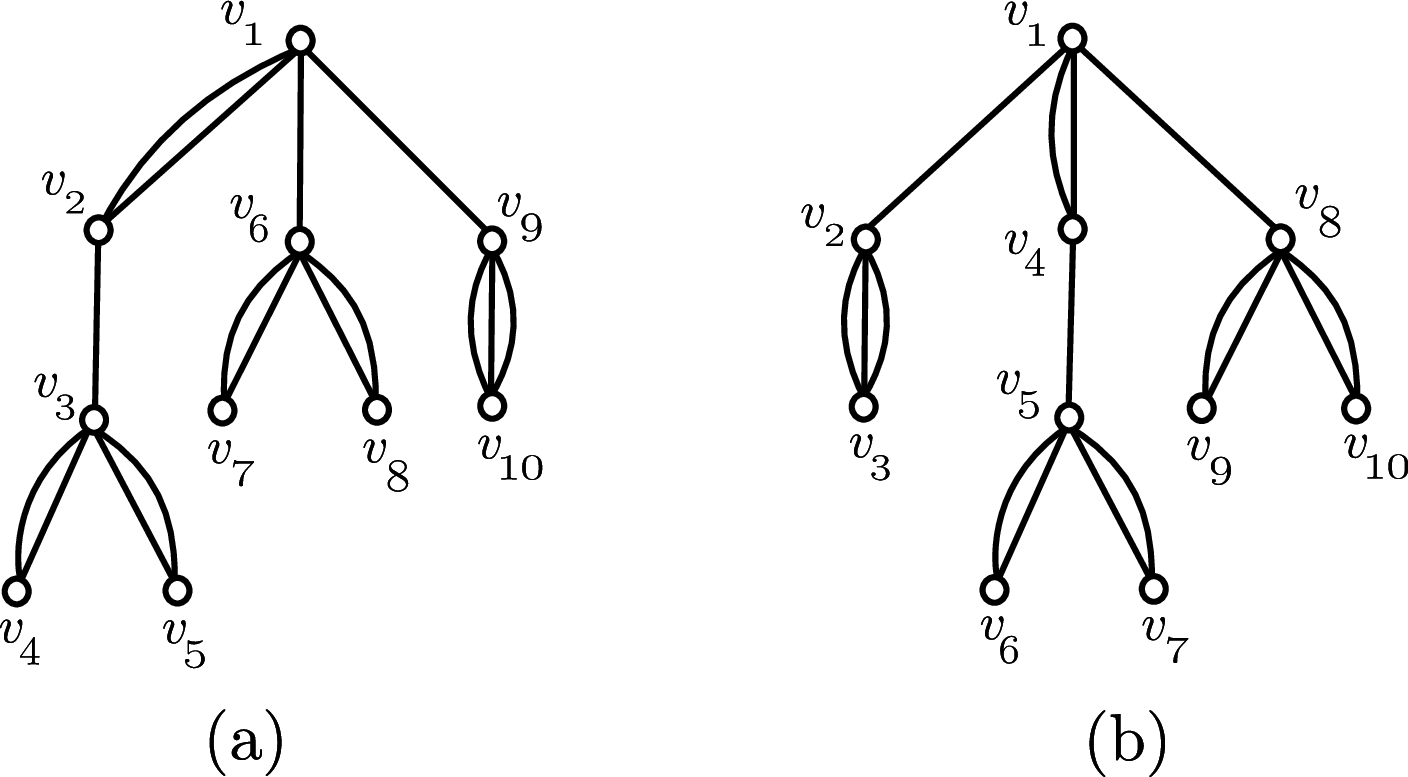}
  \caption{Representations of a graph $M$ with 10 vertices and 7 multiple edges: (a)~Canonical  representation of $M$; 
  (b)~Non-canonical representation of $M$.}
\label{fig:Tree_Like_Multigraphs}
\end{figure}

\noindent
\subsection{Subproblem} We define the subproblems for DP based on the descendant subgraphs with maximum numbers of vertices and multiple edges.  
For this purpose we define the following terms for each $M \in \mathcal{M}(n, \Delta)$, 
\begin{align*}
{\rm Max}_{\rm v}(M) & \triangleq 
\max\{\{ {\rm v}(M_{v}): v\in N(r_{M})\}\cup \{0\}\},\\
{\rm Max}_{\rm m}(M) & \triangleq 
\max \{\{{\rm e}(M_{v}): v\in N(r_{M}), {\rm v}(M_{v}) ={\rm Max}_{\rm v}(M)\} \cup\{0\}\},\\
{\rm Max}_{\rm l}(M) & \triangleq \max \{\{{\rm e}(r_{M} r_{M_{v}}): {\rm v}(M_{v})={\rm Max}_{\rm v}(M)\ \text{ and } {\ {\rm e}(M_{v})} ={\rm Max}_{\rm m}(M)\} \cup \{0\}\}.
\end{align*}
Intuitively, ${\rm Max}_{\rm v}(M)$ denotes the maximum number of vertices in the descendant subgraphs. 
${\rm Max}_{\rm m}(M)$ denotes the maximum number of multiple edges in the descendant subgraphs containing ${\rm Max}_{\rm v}(M)$ vertices. Finally,
${\rm Max}_{\rm l}(M)$ denote the maximum number of multiple edges between the root of $M$ and the roots of the descendant subgraphs containing  ${\rm Max}_{\rm v}(M)$ vertices and ${\rm Max}_{\rm m}(M)$ multiple edges. 
%Based on our ordering, a rooted tree-like multigraph view like in 
Note that for any multigraph $M \in \mathcal{M}(1, \Delta)$, it holds that  
${\rm Max}_{\rm v}(M) = 0$, $ {\rm Max}_{\rm m}(M) = 0$ and ${\rm Max}_{\rm l}(M) = 0$.
For the example multigraph in Fig.~\ref{fig:Tree_Like_Multigraphs}(a), 
${\rm Max}_{\rm v}(M) = 4$, 
${\rm Max}_{\rm m}(M) = 2$, and 
${\rm Max}_{\rm l}(M) = 1$. 

%Based on these values, we define subfamilies of $\mathcal{M}(n, \Delta)$ as follows to define subproblems for DP. 
Let $k,d,\ell$ be any three integers such that $k \geq 1$ and $d, \ell \geq 0$. 
We define a subproblem of $\mathcal{M}(n, \Delta)$ for DP as follows:
\begin{align*}
&\mathcal{M}(n, \Delta ,{k}_{\leq}, {d}_{\leq}, {\ell}_{\leq})\triangleq 
\{M\in \mathcal{M}(n, \Delta) :
{\rm Max}_{\rm v}(M) \leq k,\ {\rm Max}_{\rm m}(M)\leq d, \ 
{\rm Max}_{\rm l}(M)\leq \ell\}.
\end{align*}%
Note that by the definition of $\mathcal{M}(n, \Delta ,{k}_{\leq}, {d}_{\leq}, {\ell}_{\leq})$ it holds that 
\begin{enumerate}
[label=\textnormal{(\roman*)}, ref=(\roman*), font=\upshape]
\item $\mathcal{M}(n, \Delta ,{k}_{\leq}, {d}_{\leq}, {\ell}_{\leq}) = \mathcal{M}(n, \Delta, {n-1}_{\leq}, {d}_{\leq}, {\ell}_{\leq})$ if $k \geq n$;
\item $\mathcal{M}(n, \Delta ,{k}_{\leq}, {d}_{\leq}, {\ell}_{\leq}) =
\mathcal{M}(n, \Delta ,{k}_{\leq}, {\Delta}_{\leq}, {\ell}_{\leq})$ if $d \geq{ \Delta + 1}$; 
\item $\mathcal{M}(n, \Delta , {k}_{\leq}, {d}_{\leq}, {\ell}_{\leq})=
\mathcal{M}(n, \Delta, {k}_{\leq}, {d}_{\leq}, {\Delta}_{\leq})$ if $\ell \geq {\Delta+1}$; and
\item $\mathcal{M}(n, \Delta)= \mathcal{M}(n, \Delta, {n-1}_{\leq}, {\Delta}_{\leq}, {\Delta}_{\leq})$.
\end{enumerate}
Consequently, we assume that $k \leq n - 1$, and $ d + \ \ell \leq \Delta$.
Moreover, by the definition it holds that 
$\mathcal{M}(n, \Delta, {k}_\leq, {d}_\leq, {\ell}_\leq) \neq \emptyset$ if $``n=1$ and $\Delta=0"$ or $``n-1 \geq k \geq 1"$ and $\mathcal{M}(n, \Delta, {k}_\leq, {d}_\leq, {\ell}_\leq) = \emptyset$ 
if $``n=1$ and $\Delta \geq 1"$ or $``n \geq2$ and $k = 0"$.

We define the subproblems of $\mathcal{M}(n, \Delta, {k}_\leq, {d}_\leq, {\ell}_\leq) $ based on the maximum number of vertices in the descendant subgraphs:
\begin{equation*}
\mathcal{M}(n, \Delta, {k}_{=}, {d}_{\leq}, {\ell}_{\leq})\triangleq 
\{M\in \mathcal{M}(n, \Delta, {k}_{\leq}, {d}_{\leq}, {\ell}_{\leq}) :
{\rm Max}_{\rm v}(M) = k\}.
\end{equation*}%
\noindent 
It follows from the definition of $\mathcal{M}(n, \Delta, {k}_{=}, {d}_{\leq}, {\ell}_{\leq})$ that $\mathcal{M}(n, \Delta, {k}_{=}, {d}_{\leq}, {\ell}_{\leq}) \neq \emptyset$ if $``n = 1$ and $\Delta = 0"$ or $``n-1 \geq k \geq 1"$, and $\mathcal{M}(n, \Delta, {k}_{=}, {d}_{\leq}, {\ell}_{\leq}) = \emptyset$ if $``n=1$ and $\Delta \geq 1"$ or $`` n \geq2$ and $k = 0."$ In addition, we have the following recursive relation:
\begin{align}
&\mathcal{M}(n, \Delta, {k}_{\leq}, {d}_{\leq}, {\ell}_{\leq})= 
\mathcal{M}(n, \Delta, 0_{=}, {d}_{\leq}, {\ell}_{\leq}) \text{ if $k = 0$}, \label{equ1} \\ 
&\mathcal{M}(n, \Delta, {k}_{\leq}, {d}_{\leq}, {\ell}_{\leq})= 
\mathcal{M}(n, \Delta, {k-1}_{\leq}, {d}_{\leq}, {\ell}_{\leq})\cup \mathcal{M}(n, \Delta, {k}_{=}, {d}_{\leq}, {\ell}_{\leq})  
\text { if $k \geq 1$}, \label{equ1p}
%\begin{cases}
%\mathcal{C}n, \Delta ,0_{=},d_{\leq}) & \text{if $k = 0$},\\
%\mathcal{C}n\Delta,m-1_{\leq},d_{\leq})\cup 
%\mathcal{C}n\Delta ,k_{=},d_{\leq}) & 
%\text {otherwise (if $k \geq 1$)},
%\end{cases}
%\end{split}
\end{align}%
where 
$\mathcal{M}(n, \Delta, {k-1}_{\leq}, {d}_{\leq}, {\ell}_{\leq})\cap 
\mathcal{M}(n, \Delta, {k}_{=}, {d}_{\leq}, {\ell}_{\leq}) = \emptyset$ for $k \geq 1$. \comb{This recursive relation is shown in Fig.}~\ref{fig:dem}.
\begin{figure}[t!]
   \centering  
   \includegraphics [width=0.9\textwidth]{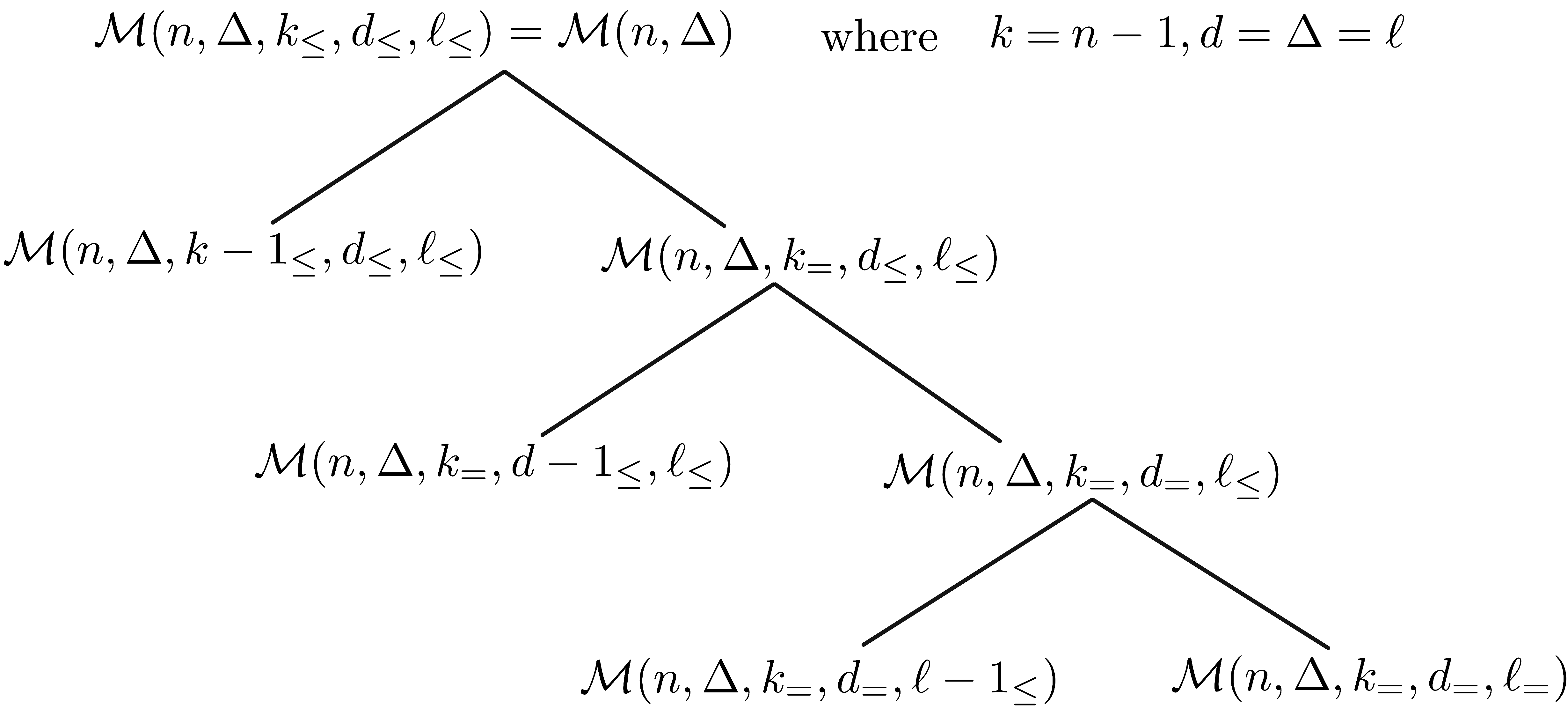}
  \caption{\comb{Illustration of the recursive relations as discussed in Eqs.~(\ref{equ1})-(\ref{equ3p}).}}
\label{fig:dem}
\end{figure}

We define the subproblems of $ \mathcal{M}(n, \Delta, {k}_{=}, {d}_{\leq}, {\ell}_{\leq})$ based on the maximum number of multiple edges in the descendant subgraphs:
\begin{equation*}
\mathcal{M}(n, \Delta, {k}_{=}, {d}_{=}, {\ell}_{\leq} )\triangleq 
\{ M\in 
\mathcal{M}(n, \Delta, {k}_{=}, {d}_{\leq}, {\ell}_{\leq}) :
{\rm Max}_{\rm m}(M)= d\}.
\end{equation*}%
By the definition of $\mathcal{M}(n, \Delta, {k}_{=}, {d}_{=}, {\ell}_{\leq})$, it holds that $\mathcal{M}(n, \Delta, {k}_{=}, {d}_{=}, {\ell}_{\leq}) \neq \emptyset$ if $``n = 1$ and $\Delta = 0"$ or $ ``n-1 \geq k \geq 1"$ and $\mathcal{M}(n, \Delta, {k}_{=}, {d}_{=}, {\ell}_{\leq}) = \emptyset$ in the case of $``n = 1$ and $\Delta \geq 1"$ or $`` n \geq2$ and $k = 0"$.\\
We also obtain the following relation for
$\mathcal{M}(n, \Delta, {k}_{=}, {d}_{\leq}, {\ell}_{\leq})$:
\begin{align}
&\mathcal{M}(n, \Delta, {k}_{=}, {d}_{\leq}, {\ell}_{\leq})= \mathcal{M}(n, \Delta, {k}_{=}, 0_{=}, {\ell}_{\leq})  
\text{ if $d = 0$}, \label{equ2} \\ 
&\mathcal{M}(n, \Delta, {k}_{=}, {d}_{\leq}, {\ell}_{\leq})= 
\mathcal{M}(n, \Delta, {k}_{=}, {d-1}_{\leq}, {\ell}_{\leq})\cup \mathcal{M}(n, \Delta, {k}_{=}, {d}_{=}, {\ell}_{\leq})  \text{ if $d \geq 1$}, \label{equ2p}
%\begin{cases}
%\mathcal{C}n\Delta ,k_{=}, 0_{=}, \ell_{\leq}) & \text{if $d = 0$},\\
%\mathcal{C}n\Delta ,k_{=}, d-1_{\leq}, \ell_{\leq})\cup  
%\mathcal{C}n\Delta ,k_{=}, d_{=}, \ell_{\leq}) & \text{otherwise (if $d \geq 1$)},
%\end{cases}
\end{align}
where $\mathcal{M}(n, \Delta, {k}_{=}, {d-1}_{\leq}, {\ell}_{\leq})\cap \mathcal{M}(n, \Delta, {k}_{=}, {d}_{=}, {\ell}_{\leq}) = \emptyset$ for $d \geq 1$. \comb{This recursive relation is shown in Fig.}~\ref{fig:dem}.\\

Similarly, we define the subproblems of 
$\mathcal{M}(n, \Delta, {k}_{=}, {d}_{=}, {\ell}_{\leq})$ based on the 
maximum number of multiple edges between the root and its children in the multigraphs:
\begin{equation*}
\mathcal{M}(n, \Delta, {k}_{=}, {d}_{=}, {\ell}_{=})\triangleq 
\{M\in \mathcal{M}(n, \Delta, {k}_{=}, {d}_{=}, {\ell}_{\leq}) :
{\rm Max}_{\rm l}(M) = \ell\}.
\end{equation*}%
\noindent By the definition of $\mathcal{M}(n, \Delta, {k}_{=}, {d}_{=}, {\ell}_{=})$, it holds that $\mathcal{M}(n, \Delta, {k}_{=}, {d}_{=}, {\ell}_{=}) \neq \emptyset$ if $``n= 1$ and $\Delta=0"$ or $``n-1 \geq k \geq 1"$ and $\mathcal{M}(n, \Delta, {k}_{=}, {d}_{=}, {\ell}_{=}) = \emptyset$ if $``n=1$ and $\Delta \geq 1"$ or $`` n \geq2$ and $k = 0"$.\\
We obtain the following relation for $\mathcal{M}(n, \Delta, {k}_{=}, {d}_{=}, {\ell}_{\leq})$:
\begin{align}
&\mathcal{M}(n, \Delta, {k}_{=}, {d}_{=}, {\ell}_{\leq})= \mathcal{M}(n, \Delta, {k}_{=}, {d}_{=}, 0_{=}) 
\text{ if $\ell= 0$}, \label{equ3} \\
&\mathcal{M}(n, \Delta, {k}_{=}, {d}_{=}, {\ell}_{\leq})= 
\mathcal{M}(n, \Delta, {k}_{=}, {d}_{=}, {\ell-1}_{\leq})\cup \mathcal{M}(n, \Delta, {k}_{=}, {d}_{=}, {\ell}_{=})  \text{ if $\ell\geq 1$}, \label{equ3p}
%\begin{cases}
%\mathcal{C}n\Delta ,k_{=}, d_{=}, 0_{=}) & \text{if \ell = 0$},\\
%\mathcal{C}n\Delta ,k_{=}, d_{=}, \ell-1_{\leq})\cup  
%\mathcal{C}n\Delta ,k_{=}, d_{=}, \ell_{=}) & \text{otherwise (if \ell \geq 1$)},
%\end{cases}
\end{align}
where $\mathcal{M}(n, \Delta, {k}_{=}, {d}_{=}, {\ell-1}_{\leq})\cap \mathcal{M}(n, \Delta, {k}_{=}, {d}_{=}, {\ell}_{=}) = \emptyset$ for $\ell \geq 1$.
\comb{This recursive relation is shown in Fig.}~\ref{fig:dem}.
%%%%%%%%
\subsection{Recursive relations} %\label{Conditions_for_Subtree}\\
To obtain recursive relations for computing the size of 
$\mathcal{M}(n, \Delta, {k}_{=}, {d}_{=}, {\ell}_{=})$, it is essential to 
analyze the types of descendant subgraphs. 
Let $M$ be a multigraph such that $M \in \mathcal{M}(n, \Delta, {k}_=, {d}_=,  {\ell}_=)$. 
Then it is easy to observe that for any vertex $v\in N(r_{M})$, the descendant subgraph $M_{v}$ satisfies exactly one of the following four conditions:

\begin{enumerate}[label=\textnormal{(\arabic*)}, ref=(\arabic*)]

\item \label{con:condit11}$ {\rm v}(M_{v}) = k$, ${\rm e}(M_{v}) = d$ and ${\rm e}({r_{M} r_{M_{v}}}) = \ell$. 

\item \label{con:condit12}$ {\rm v}(M_{v}) = k$, ${\rm e}(M_{v}) = d$ and $0 \leq {\rm e}({r_{M} r_{M_{v}}}) < \ell$. 

\item \label{con:condit13}$ {\rm v}(M_{v}) = k$, $0 \leq {\rm e}(M_{v}) < d$ and $0 \leq {\rm e}({r_{M} r_{M_{v}}}) \leq \Delta$.

\item \label{con:condit14}${\rm v}(M_{v}) < k$, $0 \leq {\rm e}(M_{v}) \leq \Delta$  and $0 \leq {\rm e}({r_{M} r_{M_{v}}}) \leq \Delta$.
\end{enumerate}

These cases are demonstrated by an example of multigraph $M \in  \mathcal{M}(20, 13)$  in Fig.~\ref{fig:subtrees_conditions}.

\begin{figure}[h!]
    \centering
    \includegraphics[width=0.45\textwidth]{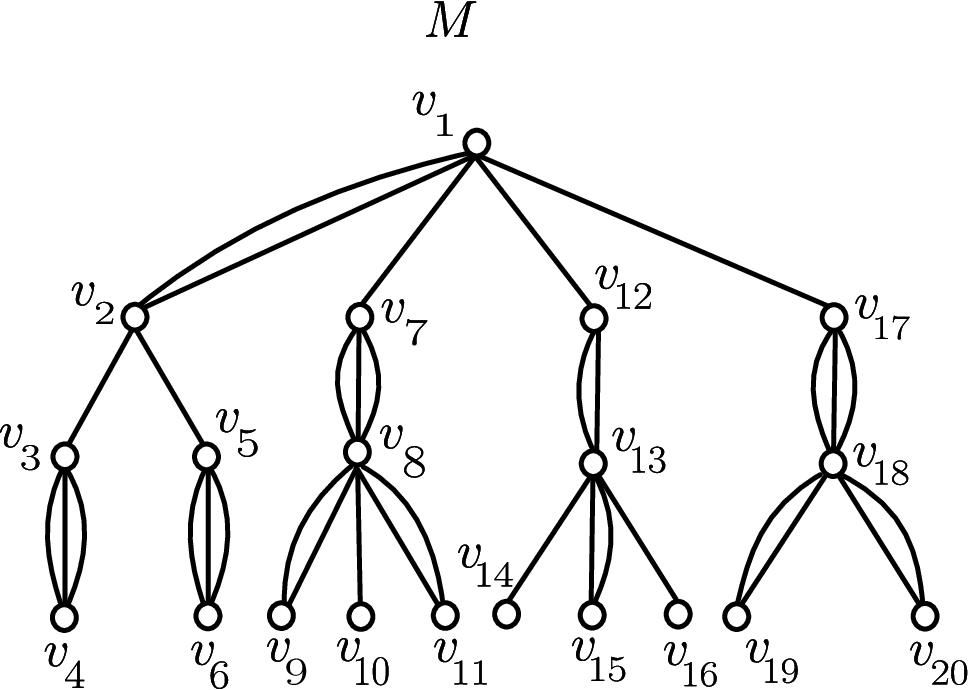}
    \caption{The descendant subgraphs rooted at $v_2$, $v_7$, $v_{12}$, and $v_{17}$ have following properties: for $v_2$, 
~$\left\vert {\rm v}(M_{v_2})\right\vert = 5$, ${\rm e}(M_{v_2})= 4$ and ${\rm e}({r_{M}} r_{M_{v_2}})$ $=$ $1$; for $v_7$, $\left\vert {\rm v}(M_{v_7})\right\vert =5$, ${\rm e}(M_{v_7})=4$ and
${\rm e}({r_{M} r_{M_{v_7}}}) = 0$; for $v_{12}$, $\left\vert {\rm v}(M_{v_{12}})\right\vert =5$, ${\rm e}(M_{v_{12}}) = 2$ and ${\rm e}({r_{M} r_{M_{v_{12}}}})= 0$; for $v_{17}$, $\left\vert {\rm v}(M_{v_{17}})\right\vert = 4, {\rm e}(M_{v_{17}}) = 4$ and ${\rm e}({r_{M} r_{M_{v_{17}}}}) = 0.$}
    \label{fig:subtrees_conditions}
\end{figure} 

We define the {\it residual multigraph} of $M \in \mathcal{M}(n, \Delta, {k}_=, {d}_=, {\ell}_=)$ to be the multigraph rooted at $r_M$ that is induced by the vertices 
${V}(M)\setminus \bigcup\limits_{\mathclap{\substack{v\in N(r_{M}),\\ 
 M_v \in \mathcal{M}(k, d, k-1_{\leq}, d_{\leq}, d_{\leq})}}}{ V}(M_v).$ 
 Basically, residual multigraph consists of those descendant subgraphs of $M$ whose structures are not clear. 
The residual multigraph of $M$ contains at least one vertex, which is the root of $M$. 
An illustration of a residual multigraph is given in Fig.~\ref{fig:residual_trees}. 

\begin{figure}[h!]
    \centering
    \includegraphics[width=0.35\textwidth]{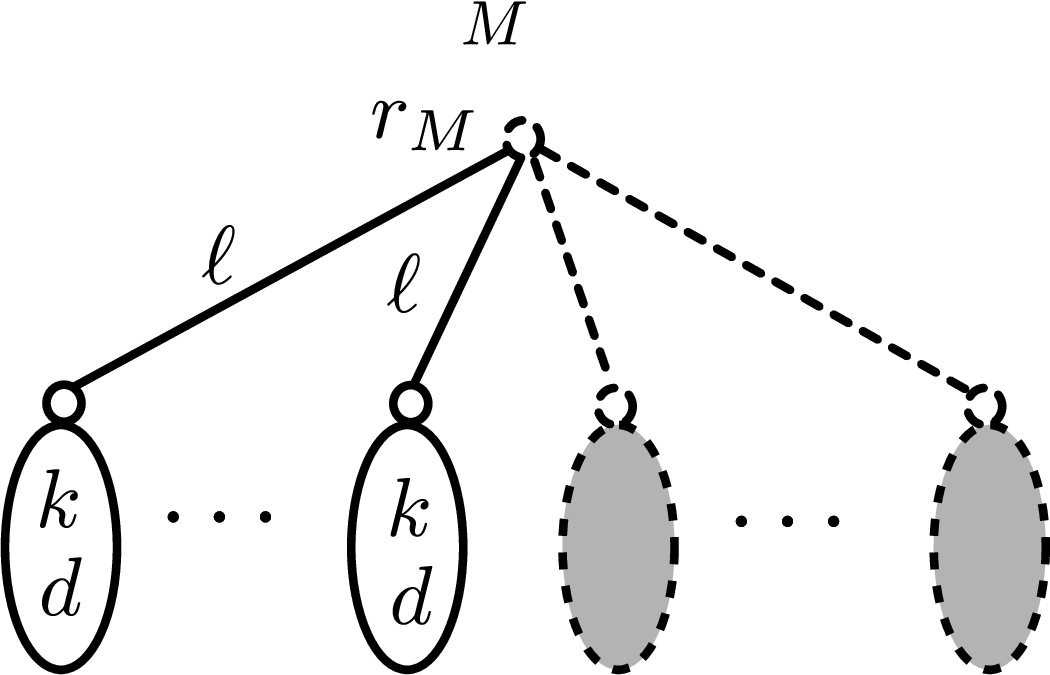}
    \caption{\comb{An illustration of the residual multigraph depicted with dashed line, and the descendant subgraph is depicted with solid lines. The descendant subgraph contains $k$ vertices, $d$ multiple edges, and $\ell$ multiple edges between the roots.}}
    \label{fig:residual_trees}
\end{figure}

A characterization of residual multigraph is discussed in Lemma~\ref{L_1}. 
\begin{lemma}
\label{L_1}
For any five integers $n \geq 3$, $k \geq 1$, $\Delta \geq d+\ell \geq 0$, and a multigraph $M\in 
\mathcal{M}(n, \Delta, {k}_{=}, {d}_=, {\ell}_=),$ let~$p = |\{ v\in N(r_{M}) :
 M_v \in \mathcal{M}(k, d, {k-1}_{\leq}, {d}_{\leq}, {d}_{\leq})\}|$. 
 Then it holds that
\begin{enumerate}[label=(\roman*), ref  = (\roman*), font=\upshape]
\item $1\leq  p \leq \lfloor (n-1)/k \rfloor$ with 
$p \leq \lfloor \Delta / (d+\ell) \rfloor $ when $d+\ell \geq 1$. \label{L_1_i}

\item \label{L_1_ii} The residual multigraph of $M$ belongs to exactly one of the following families:\\
$\mathcal{M}(n~-~pk, \Delta~-~p(d~+~\ell), {k}_=, {d}_=, \min{\{\Delta-p(d+\ell), \ell-1\}}_{\leq})$;\\ 
$\mathcal{M}(n~-~pk, \Delta~-~p(d~+~\ell), {k}_=, \min{\{\Delta-p(d+\ell), d-1\}}_{\leq}, \Delta~-~p(d~+~\ell)_{\leq})$; or \\
$
\mathcal{M}(n~-~pk, \Delta~-~p(d~+~\ell), \min\{n-pk-1, k-1\}_{\leq}, \Delta-p(d+\ell)_{\leq}, \Delta-p(d+\ell)_{\leq}).$
\end{enumerate}
\end{lemma}
\begin{proof}
\begin{enumerate}[label=\textnormal{(\roman*)}, ref=(\roman*), font=\upshape]
\item Since $M \in \mathcal{M}(n, \Delta, {k}_=, {d}_=, {\ell}_=)$, 
 there exists at least one vertex $v \in N(r_M)$ such that  
 $M_v \in \mathcal{M}(k, d, {k-1}_{\leq}, {d}_{\leq}, {d}_{\leq})$. 
It follows that $p \geq 1$. 
It further asserts that $n-1 \geq pk$ and $\Delta \geq p(d~+~\ell)$. 
 This concludes that $p \leq \lfloor (n-1)/k \rfloor$ with 
$p \leq \lfloor \Delta / (d+\ell) \rfloor $ when $d~+~\ell \geq 1$.
\item  Let $R$ represents the residual multigraph of $M$. 
By the definition of $R$, it holds that 
$ R \in \mathcal{M}(n~-~pk,~ \Delta-p(d+\ell),~ n-pk-1_{\leq},~ {\Delta-p{(d+\ell)}}_{\leq},~ {\Delta-  p(d+\ell)}_{\leq})$.  
Moreover, for each vertex $v\in N(r_{M})\cap V(R),$ the descendant subgraph 
$M_{v}$ satisfies exactly one of the Conditions~\ref{con:condit12}-\ref{con:condit14}.
Now, if there exists a vertex $v\in N(r_{M})\cap V(R)$ such that $M_v$ satisfies Condition~\ref{con:condit12} as illustrated in Fig.~\ref{fig:Multigraphs}(a), then $\ell-1\geq 0$, and hence $R\in \mathcal{M}(n~-~pk, \Delta~-~p(d~+~\ell), {k}_=, {d}_=, \min\{\Delta-p(d+\ell), \ell~-~1\}_{\leq})$.
For $v\in N(r_{M})\cap V(R)$ such that $M_v$ satisfies Condition~\ref{con:condit13} as illustrated in Fig.~\ref{fig:Multigraphs}(b) i.e., ${\rm v}(M_v) = k$, $0 \leq {\rm e}(M_v) \leq \min\{\Delta-p(d+\ell), d-1\} $ and $ 0 \leq {\rm e}({r_{M} r_{M_{v}}}) \leq \Delta-p(d+\ell)$, then the residual multigraph  $R\in \mathcal{M}(n-pk, \Delta-p(d+\ell), {k}_=, \min\{\Delta-p(d+\ell), d-1\}_{\leq}, \Delta-p(d+\ell)_{\leq})$.
If $M_v$ satisfies Condition~\ref{con:condit14} as illustrated in Fig.~\ref{fig:Multigraphs}(c) i.e., ${\rm v}(M_v) < k$, $ 0 \leq {\rm e}(M_v) \leq \Delta-p(d+\ell)$ and $0 \leq {\rm e}({r_{M} r_{M_{v}}}) \leq \Delta-p(d+\ell) $ then by the definition of $R$, it holds that $R\in \mathcal{M}(n-pk, \Delta - p(d+\ell), {\min\{n-pk-1, k-1\}}_{\leq}, \Delta-p(d+\ell)_{\leq}, \Delta-p(d+\ell)_{\leq})$. 
\end{enumerate}
\begin{figure}[h!]
   \centering  
   \includegraphics [width=0.9\textwidth]{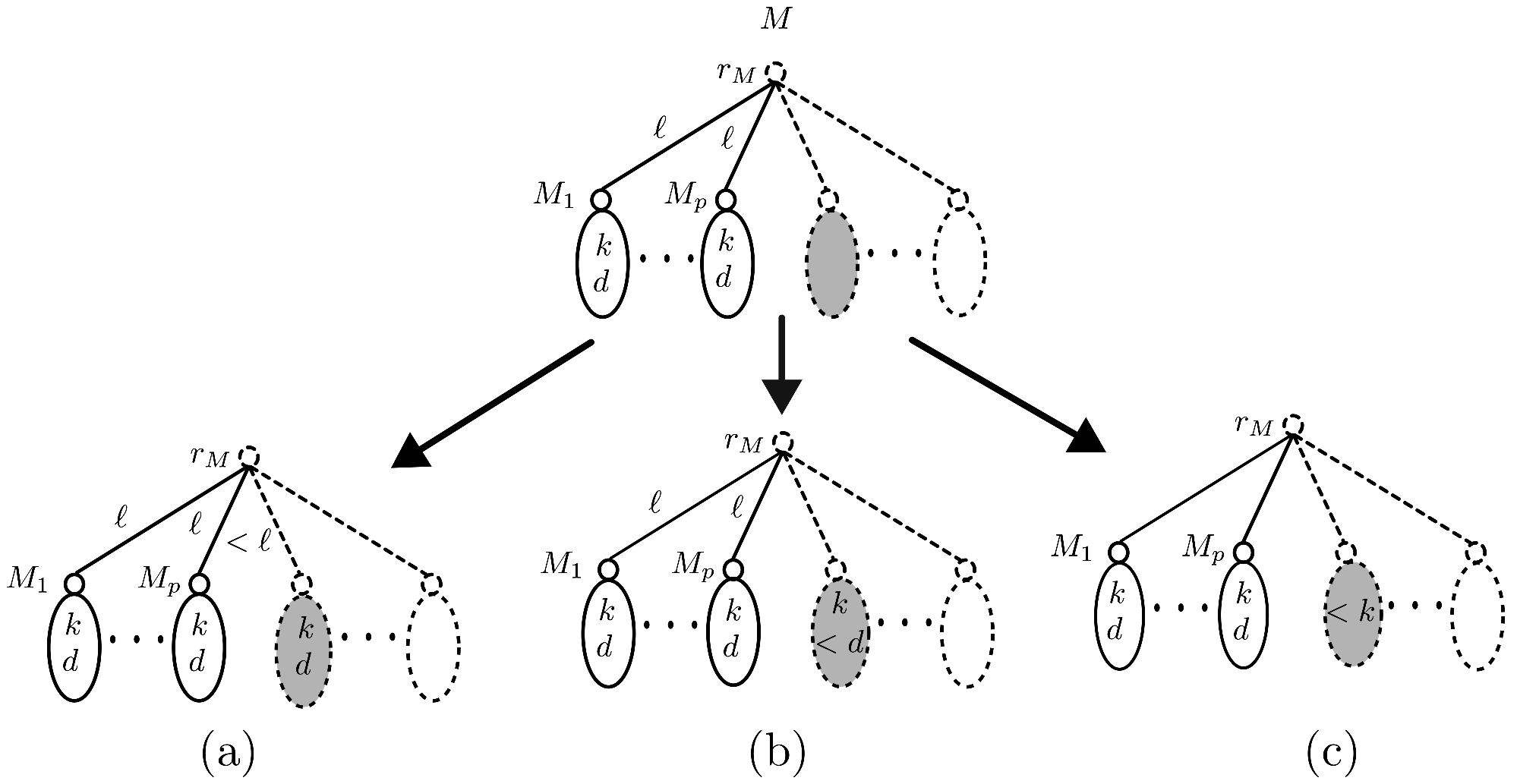}
  \caption{\comb{An illustration of different kinds of residual multigraphs depicted with dashed lines.}}
\label{fig:Multigraphs}
\end{figure}
\end{proof}

Let $n$, $k$, $\Delta$, $d$ and $\ell$ be five non-negative integers such that $ n-1 \geq k \geq 0$, and $\Delta \geq d+\ell \geq 0$. 
Let $m(n, \Delta, {k}_{\leq}, {d}_{\leq}, {\ell}_{\leq}), m(n, \Delta, {k}_=, {d}_{\leq}, {\ell}_{\leq}), m(n, \Delta, {k}_=, {d}_=, {\ell}_{\leq})$ and $m(n, \Delta, {k}_=, {d}_=, {\ell}_=)$ represent the number of elements in the families 
$\mathcal{M}(n, \Delta, {k}_{\leq}, {d}_{\leq}, {\ell}_{\leq})$, $\mathcal{M}(n, \Delta, {k}_=, {d}_{\leq}, {\ell}_{\leq})$, $\mathcal{M}(n, \Delta, {k}_=, {d}_=, {\ell}_{\leq})$ and $\mathcal{M}(n, \Delta, {k}_=, {d}_=, {\ell}_=)$, respectively. 
%
%
%\noindent \textbf{Recursive Relations} \label{Recursive_Relations}\\
For any six integers $ n\geq3$, $k \geq 1$, $\Delta \geq d+\ell \geq 0$, 
and $x \geq 0$, 
let $$f(k, d; x) \triangleq \binom{m(k, d, {k-1}_\leq, {d}_\leq, {d}_\leq)+ x - 1}{x}$$ denote the number of combinations with repetition of $x$ descendant subgraphs from the family 
$\mathcal{M}(k, d, {k-1}_\leq, {d}_\leq, {d}_\leq)$. 
With all necessary information, we are ready to discuss a recursive relation for 
$m(n, \Delta, {k}_=, {d}_=, {\ell}_=)$ in~Lemma~\ref{lem_a} \comb{ and these relations are shown in Fig.}~\ref{fig:lemma2}.
 \begin{figure}[h!]
   \centering  
   \includegraphics [width=0.45\textwidth]{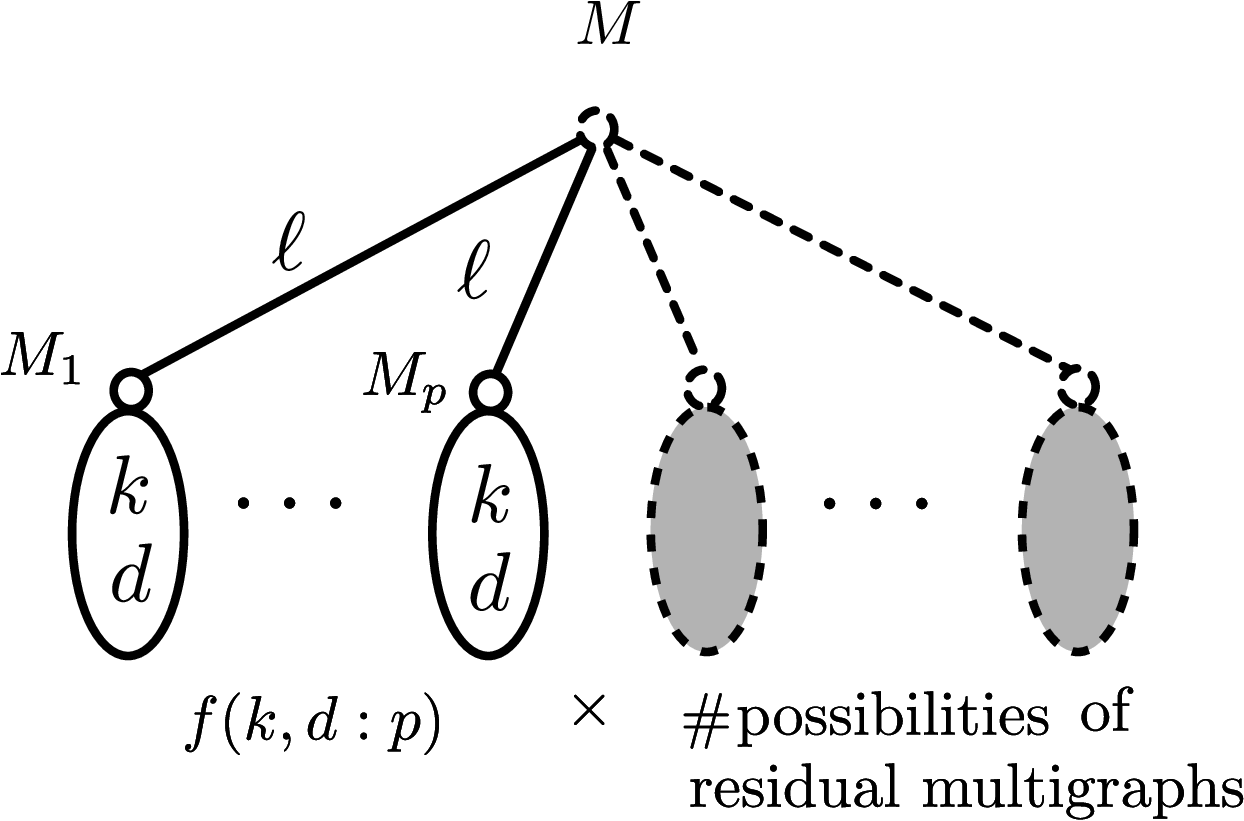}
  \caption{\comb{For fixed parameters $k$, $d$ and $p$, the number of possible descendent multigraphs $M_p$. }}
\label{fig:lemma2}
\end{figure}

\begin{lemma}
\label{lem_a}
For any six integers $ n \geq 3$, $k \geq 1$, $\Delta \geq d+\ell \geq 0$, and $p$, 
such that $1\leq p \leq \lfloor (n - 1)/k \rfloor$ with 
$p \leq \lfloor \Delta / (d~+~\ell) \rfloor$ when $d+\ell \geq 1 $ and $0\leq \ell \leq \Delta$, it holds that 
\begin{enumerate}[label=(\roman*), ref  = (\roman*), font=\upshape]
\item\label{lem_a_i}
$m(n, \Delta, {k}_=, {d}_=, {\ell}_=) = \sum_{p}f(k, d; p) (m(n-pk, \Delta, \min{\{n-pk-1, k-1\}}_{\leq},\\ {\Delta}_{\leq}, {\Delta}_{\leq}))$ if $d = 0$ and $\ell = 0$;

\item\label{lem_a_ii}
$m(n, \Delta, {k}_=, {d}_=, {\ell}_=) = \sum_{p}f(k, d; p)(m(n-pk, \Delta-pd,  
 {k}_=, \min{\{\Delta-pd, 
 d-1\}}\\_{\leq}, {\Delta-pd}_{\leq}) +   m(n-pk, \Delta-pd,  \min{\{n-pk-1, k-1\}}_{\leq}, {\Delta-pd}_{\leq}, {\Delta-pd}_{\leq}))$ if $d \geq 1$ and $\ell = 0$;

\item\label{lem_a_iii}
$m(n, \Delta, {k}_{=}, {d}_{=}, {\ell}_{=}) = \sum_{p}f(k, d; p) (m(n-pk, \Delta-p\ell, {k}_{=}, {d}_{=}, \\ \min{\{\Delta-p\ell, \ell-1\}}_{\leq}) + m(n-pk, \Delta-p\ell, \min{\{n-pk-1, k-1\}}_{\leq}, {\Delta-p\ell}_{\leq},\\ {\Delta-p\ell}_{\leq}))$ if $ d = 0$ and $\ell \geq 1$;

\item\label{lem_a_iv}
$m(n, \Delta, {k}_=, {d}_=, {\ell}_=) = \sum_{p}f(k, d; p) (m(n-pk, \Delta-p(d+\ell), {k}_=, {d}_=,\\  \min{\{\Delta-p(d+\ell), \ell-1\}}_{\leq})+m(n-kp, \Delta-p(d+\ell), {k}_=, \\ \min{\{\Delta-p(d+\ell), d-1\}}_{\leq}, {\Delta-p(d+\ell)}_{\leq})+m(n-kp, \Delta-p(d+\ell), \\ \min{\{n-pk-1, k-1\}}_{\leq}, {\Delta-p(d+\ell)}_{\leq},  {\Delta-p(d+\ell)}_{\leq}))$ if $d \geq 1$ and $\ell \geq 1.$ 

\end{enumerate}
\end{lemma}
\begin{proof}
Let us consider a multigraph $M$ from the family $\mathcal{M}(n, \Delta, {k}_=, {d}_=, {\ell}_=)$. Since $m(k, d, {k-1}_{\leq}, {d}_\leq, {d}_\leq)$ represents the number of multigraphs in the family $\mathcal{M}(k, d, {k-1}_{\leq}, {d}_{\leq}, {d}_{\leq})$. By Lemma~\ref{L_1}(i), for a specific value of $p$ as $1\leq p \leq \lfloor (n~-~1)/k \rfloor$ with 
$p \leq \lfloor \Delta/(d+\ell) \rfloor $ when $d+\ell \geq 1$, there are exactly $p$ descendant subgraphs $M_v$ for $v\in N(r_M)$ and 
$M_v \in \mathcal{M}(k, d, {k-1}_{\leq}, {d}_{\leq}, {d}_{\leq})$. 
As $ \binom{m(k, d, {k-1}_\leq, {d}_\leq, {d}_\leq)+ p-1}{p}$ 
denote the number of combinations with repetition of $p$ trees from the family 
$\mathcal{M}(k, d, {k-1}_\leq, {d}_\leq, {d}_\leq)$ for $v\in N(r_{M})$.
Further, by Lemma~\ref{L_1}(ii) 
the residual subgraph $R$ of $M$ belongs to exactly one of the following families,\\

\noindent $\mathcal{M}(n-pk, \Delta, \min{\{n-pk-1, k-1\}}_{\leq}, {\Delta}_{\leq}, {\Delta}_{\leq}))$ if $d = 0$ and $\ell = 0$; or\\

\noindent$\mathcal{M}(n-pk, \Delta-pd, {k}_=, \min{\{\Delta-pd, d-1\}}_{\leq}, {\Delta-pd}_{\leq}) \ \cup \ \mathcal{M}(n-pk, \Delta-pd, $\\$ \min{\{n-pk-1, k-1\}}_{\leq}, {\Delta-pd}_{\leq}, {\Delta-pd}_{\leq}))$ if $d \geq 1$ and $\ell = 0$; or\\

\noindent$\mathcal{M}(n-pk, \Delta-p\ell, {k}_=, {d}_=, \min{\{\Delta-p\ell, \ell-1\}}_{\leq})\ \cup \ \mathcal{M}(n-pk, \Delta-p\ell, \\ \min{\{n-pk-1, k-1\}}_{\leq}, {\Delta-p\ell}_{\leq}, {\Delta-p\ell}_{\leq}))$ if $d = 0$ and $\ell \geq 1$;\\

\noindent$\mathcal{M}(n-pk, \Delta-p(d+\ell), {k}_=, {d}_=, \min{\{\Delta-p(d+\ell), \ell-1\}}_{\leq}) \cup \mathcal{M}(n-pk, \Delta-p(d+\ell), {k}_=, \min{\{\Delta-p(d+\ell), d-1\}}_{\leq}, {\Delta-p(d+\ell)}_{\leq}) \cup \mathcal{M}(n-pk, \Delta-p(d+\ell), \min{\{n-pk-1, k-1\}}_{\leq},  {\Delta-p(d+\ell)}_{\leq}, {\Delta-p(d+\ell)}_{\leq}))$ if $d \geq 1$ and $\ell \geq 1.$ \\

Note that in each case, intersection of any two families of the residual trees will always result in an empty set.
By the sum rule of counting, the total number of multigraphs in the families of residual multigraphs are
\begin{enumerate}[label=(\alph*), ref  = (\alph*), font=\upshape]
\item \label{a} $m(n-pk, \Delta, \min{\{n-pk-1, k-1\}}_{\leq}, {\Delta}_{\leq}, {\Delta}_{\leq})$ if $d = 0$ and $\ell = 0$; 
\item \label{b} $m(n-pk, \Delta-pd, {k}_=, \min{\{\Delta-pd, d-1\}}_{\leq}, {\Delta-pd}_{\leq})+m(n-pk, \Delta-pd, $\\$ \min{\{n-pk-1, k-1\}}_{\leq}, {\Delta-pd}_{\leq}, {\Delta-pd}_{\leq})$ if $d \geq 1$ and $\ell = 0$; 
\item \label{c} $m(n-pk, \Delta-p\ell, {k}_=, {d}_=, \min{\{\Delta-p\ell, \ell-1\}}_{\leq})+m(n-pk, \Delta-p\ell, $\\$ \min{\{n-pk-1, k-1\}}_{\leq}, {\Delta-p\ell}_{\leq}, {\Delta-p\ell}_{\leq})$ if $d = 0$ and $\ell \geq 1$;
\item \label{d} $m(n-pk, \Delta-p(d+\ell), {k}_=, {d}_=, \min{\{\Delta-p(d+\ell), \ell-1\}}_{\leq})+m(n-pk, \Delta-p(d+\ell), {k}_=, \min{\{\Delta-p(d+\ell), d-1\}}_{\leq}, {\Delta-p(d+\ell)}_{\leq})+m(n-pk, \Delta-p(d+\ell), \min{\{n-pk-1, k-1\}}_{\leq}, {\Delta-p(d+\ell)}_{\leq}, {\Delta-p(d+\ell)}_{\leq})$ if $d \geq 1$ and $\ell \geq 1.$ \\
\end{enumerate}
 This implies that for a fixed integer $p$ in the range given in the statement, and by the product rule of counting, the number of residual multigraphs $R$ in the family $\mathcal{M}(n, \Delta, {k}_=, {d}_=, {\ell}_=)$  
with exactly $p$ descendant subgraphs 
$R_v \in \mathcal{M}(k, d, {k-1}_{\leq }, {d}_{\leq}, {d}_{\leq})$, for $v \in M(r_R)$, are 
\smallskip
\begin{enumerate}[label=(\alph*), ref  = (\alph*), font=\upshape]
\item \label{a}  $f(k, d; p)
(m(n-pk, \Delta, \min{\{n-pk-1, k-1\}}_{\leq}, {\Delta}_{\leq}, {\Delta}_{\leq}))$ if $d = 0$ and $\ell = 0$; 
\item\label{b} $f(k, d; p) 
(m(n-pk, \Delta-pd, {k}_=, \min{\{\Delta-pd, d-1\}}_{\leq}, {\Delta-pd}_{\leq})+m(n-pk, \Delta-pd, \min{\{n-pk-1, k-1\}}_{\leq}, {\Delta-pd}_{\leq}, {\Delta-pd}_{\leq}))$ if $d \geq 1$ and $\ell = 0$; 
\item\label{c} $f(k, d; p) 
(m(n-pk, \Delta-p\ell, {k}_=, {d}_=, \min{\{\Delta-p\ell, \ell-1\}}_{\leq})+m(n-pk, \Delta-p\ell, \min{\{n-pk-1, k-1\}}_{\leq}, {\Delta-p\ell}_{\leq}, {\Delta-p\ell}_{\leq}))$ if $d = 0$ and $\ell \geq 1$;
\item\label{d} $f(k, d; p) 
(m(n-pk, \Delta-p(d+\ell), {k}_=, {d}_=, \min{\{\Delta-p(d+\ell), \ell-1\}}_{\leq})+m(n-pk, \Delta-p(d+\ell), {k}_=, \min{\{\Delta-p(d+\ell), d-1\}}_{\leq}, {\Delta-p(d+\ell)}_{\leq})+m(n-pk, \Delta-p(d+\ell), \min{\{n-pk-1, k-1\}}_{\leq}, {\Delta-p(d+\ell)}_{\leq}, {\Delta-p(d+\ell)}_{\leq}))$  if $d \geq 1$ and $\ell \geq 1.$ 

\end{enumerate}
\smallskip
Note that in the case of $n = 1$ and $\Delta \geq 1$,  $m(1, \Delta) = 0$, 
whereas for $n = 2$ and $\Delta \geq 1$, we have  $m(2, \Delta) = 1$. 
Similarly for $k = 1$ and $\Delta = 0$, we have $1 \leq p \leq n-1$, then by Lemma~\ref{l_1}(ii), it holds that $m(n-p, 0, 0_{\leq}, 0_{\leq}, 0_{\leq}) = 1$ if $p = n-1$ and $m(n-p, 0, 0_{\leq}, 0_{\leq}, 0_{\leq}) = 0$ if $1 \leq p \leq n-2$. This implies that any multigraph $M \in \mathcal{M}(n, 0, 1_=, 0_=, 0_=)$ has exactly $p = n - 1$ descendant subgraph $ M_v \in \mathcal{M}(1, 0, 0_{\leq }, 0_{\leq}, 0_{\leq})$, for $v\in N(r_{M})$. However, observe that for each integer $k \geq 2$ or $d+\ell \geq 1$, and $p$ satisfying  the conditions 
given in the statement, there exists at least one multigraph 
$M \in \mathcal{M}(n, \Delta, {k}_=, {d}_=, {\ell}_=)$ such that $M$ has exactly 
$p$ descendant subgraphs $M_v \in \mathcal{M}(k, d, {k-1}_{\leq}, {d}_{\leq}, {d}_{\leq})$, for 
$v \in N(r_M)$. 
This and cases~(a)-(d), imply the statements~(i)-(iv), respectively.
\end{proof}
\begin{theorem}\label{theom}
\comb{For any six integers $ n \geq 3$, $k \geq 1$, $\Delta \geq d+\ell \geq 0$, and $p$, 
such that $1\leq p \leq \lfloor (n - 1)/k \rfloor$ with 
$p \leq \lfloor \Delta / (d~+~\ell) \rfloor$ when $d+\ell \geq 1 $ and $0\leq \ell \leq \Delta$, it holds that }
\begin{enumerate}[label=(\roman*), ref  = (\roman*), font=\upshape]
\item\label{lem_a_v}
\comb{$m(n, \Delta, {k}_=, {d}_=, {\ell}_=) = \sum_{p}f(k, d; p-1) 
 (m(k, d, {k-1}_\leq, {d}_\leq, {d}_\leq) + p - 1)/ p) (m(n-pk, \Delta,  \min{\{n-pk-1, k-1\}}_{\leq}, {\Delta}_{\leq}, {\Delta}_{\leq}))$  if $d = 0$ and $\ell = 0$;}
 
 \item\label{lem_a_vi}
\comb{$m(n, \Delta, {k}_=, {d}_=, {\ell}_=) = \sum_{p}f(k, d; p-1) 
 (m(k, d, {k-1}_\leq, {d}_\leq, {d}_\leq) + p - 1)/ p) $\\$ 
(m(n-pk, \Delta - pd, {k}_=, \min{\{\Delta-pd, d-1\}}_{\leq}, {\Delta-pd}_{\leq})+m(n-pk, \Delta-pd, $\\$ \min{\{n-pk-1, k-1\}}_{\leq}, {\Delta-pd}_{\leq}, {\Delta-pd}_{\leq}))$ if $d\geq1$ and $\ell = 0$;}

\item\label{lem_a_vii}
\comb{$m(n, \Delta, {k}_=, {d}_=, {\ell}_=) = \sum_{p}f(k, d; p-1) 
 (m(k, d, {k-1}_\leq, {d}_\leq, {d}_\leq) + p - 1)/ p) 
(m \\(n-pk, \Delta-p\ell, {k}_=, {d}_=, \min{\{\Delta-p\ell, \ell-1\}}_{\leq})+m(n-pk, \Delta-p\ell,  \min{\{n-pk-1, k-1\}}_{\leq}, {\Delta-p\ell}_{\leq},  {\Delta-p\ell}_{\leq}))$ if $d = 0$ and $\ell \geq1$;}

\item\label{lem_a_viii}
\comb{$m(n, \Delta, {k}_=, {d}_=, {\ell}_=) = \sum_{p}f(k, d; p-1)(m(k, d, {k-1}_\leq, {d}_\leq, {d}_\leq)+p-1)/ p) $\\$ 
 (m(n-pk, \Delta-p(d+\ell), {k}_=, {d}_=, \min{\{\Delta-p(d+\ell), \ell-1\}}_{\leq})+m(n-pk, \Delta-p(d+\ell), {k}_=, \min{\{\Delta-p(d+\ell), d-1\}}_{\leq}, {\Delta-p(d+\ell)}_{\leq})+m(n-pk, \Delta-p(d+\ell), \min{\{n-pk-1, k-1\}}_{\leq}, {\Delta-p(d+\ell)}_{\leq}, {\Delta-p(d+\ell)}_{\leq}))$ if $d \geq 1$ and $\ell \geq 1.$ }
\end{enumerate}
\end{theorem}
\begin{proof}
The proof of Theorem~\ref{theom} follows from Lemma~\ref{lem_a}.  Furthermore, it holds that 
 \begin{align*}
      %\begin{split}
 f(k, d; p) = &\ \frac{(m(k, d, {k-1}_\leq, {d}_\leq, {d}_\leq) + p - 1)!}{(m(k, d, {k - 1}_\leq, {d}_\leq, {d}_\leq) - 1)!p!} \\
 %%%%
=& \frac{(m(k, d, {k - 1}_\leq, {d}_\leq, {d}_\leq) + p -2)!}{(m(k, d, {k - 1}_\leq, {d}_\leq, {d}_\leq) - 1)!(p - 1)!} \times
\vspace{2mm}
\frac{(m(k, d, {k - 1}_\leq, {d}_\leq, {d}_\leq) + p - 1)}{p}\\
  %%%%
  =& f(k, d; p - 1) \times \frac{(m(k, d, {k - 1}_\leq, {d}_\leq, {d}_\leq) + p - 1)}{ p}.
 %\end{split}
 \end{align*}
\end{proof}

In Lemma~\ref{first_recursion}, we discuss the recursive relations for $m(n, \Delta, {k}_{\leq}, {d}_{\leq}, {\ell}_{\leq})$, $m(n, \Delta, {k}_=, {d}_{\leq}, {\ell}_{\leq})$ and $m(n, \Delta, {k}_=, {d}_=, {\ell}_{\leq})$. 

\begin{lemma}\label{first_recursion}
For any five integers $ n-1 \geq k \geq 0$, and $\Delta \geq d+\ell \geq 0$, 
it holds that 
\begin{enumerate}[label=\textnormal{(\roman*)}, ref=(\roman*), font=\upshape]

\item \label{first_recursion_i} 
$m(n, \Delta, {k}_{\leq}, {d}_{\leq}, {\ell}_{\leq})= m(n, \Delta, 0_{=}, {d}_{\leq}, {\ell}_{\leq})$ if $k = 0$;

\item \label{first_recursion_ii} 
$m(n, \Delta, {k}_{\leq}, {d}_{\leq}, {\ell}_{\leq})= m(n, \Delta, {k-1}_{\leq}, {d}_{\leq}, {\ell}_{\leq}) + m(n, \Delta, {k}_{=}, {d}_{\leq}, {\ell}_{\leq})$  if $k \geq 1$;

\item \label{first_recursion_iii} 
$m(n, \Delta, {k}_=, {d}_{\leq}, {\ell}_{\leq})= m(n, \Delta, {k}_=, 0_=, {\ell}_{\leq}) $ if $d = 0$;

\item \label{first_recursion_iv} 
$m(n, \Delta, {k}_=, {d}_{\leq}, {\ell}_{\leq})= m(n, \Delta, {k}_=, {d-1}_{\leq}, {\ell}_{\leq}) +
m(n, \Delta, {k}_=, {d}_=, {\ell}_{\leq})$  if $d \geq 1$; 

\item \label{first_recursion_v}
$m(n, \Delta, {k}_=, {d}_=, {\ell}_{\leq})= m(n, \Delta, {k}_=, {d}_=, 0_=)$ if $\ell = 0$; and

 \item \label{first_recursion_vi}
$m(n, \Delta, {k}_=, {d}_=, {\ell}_{\leq}) = m(n, \Delta, {k}_=, {d}_=, {\ell-1}_{\leq})+
m(n, \Delta, {k}_=, {d}_=, {\ell}_=)$ if $\ell \geq 1.$
\end{enumerate}
\end{lemma}
\begin{proof}
By Eq.~(\ref{equ1}), the case~(i) follows. Similarly, by Eq.~(\ref{equ1p}), and the fact that $k \geq 1$, it holds that 
$\mathcal{M}(n, \Delta, {k-1}_{\leq}, {d}_{\leq}, {\ell}_{\leq})\cap 
\mathcal{M}(n, \Delta, {k}_{=}, {d}_{\leq}, {\ell}_{\leq}) = \emptyset$, case~(ii) follows.
Similarly by Eq.~(\ref{equ2}), case~(iii) follows. 
Similarly by Eq.~(\ref{equ2p}) and the fact that for $d \geq 1$ it holds that 
$\mathcal{M}(n, \Delta, {k}_=, {d-1}_{\leq}, {\ell}_{\leq})\cap
\mathcal{M}(n, \Delta, {k}_=, {d}_=, {\ell}_{\leq}) = \emptyset$, the case~(iv) follows.
By Eq.~(\ref{equ3}), case~(v) follows. Similarly by Eq.~(\ref{equ3p}) and the fact that. for 
$\ell \geq 1$ it holds that 
$\mathcal{M}(n, \Delta, {k}_=, {d}_=, {\ell-1}_{\leq})\cap 
\mathcal{M}(n, \Delta, {k}_{=}, {d}_=, {\ell}_=) = \emptyset$, case~(vi) follows.
\end{proof}
\subsection{Initial conditions }\label{Initial_Conditions}
In Lemma~\ref{l_1}, we discuss the initial conditions that are necessary to design a DP algorithm. 
\begin{lemma}
\label{l_1}
For any five integers $ n-1 \geq k \geq 0$, $\Delta \geq d+\ell \geq 0$, it holds that

\begin{enumerate}[label=\textnormal{(\roman*)}, ref=(\roman*), font=\upshape]
\item \label{l_1_i} $m(n, \Delta, 0_=, {d}_{=}, {\ell}_{=}) = 1$ if ``$n =1$ \text{ and } $\Delta= 0$"
and $m(n, \Delta, 0_{\leq}, {d}_{\leq}, {\ell}_{\leq})=0$ if ``$n = 1$ and $\Delta \geq 1$" or, ``$n \geq 2$''; 
\item \label{l_1_ii} $m(n, \Delta, 0_=, {d}_{\leq}, {\ell}_{\leq})=m(n,  \Delta, 0_{\leq}, {d}_{\leq},  {\ell}_{\leq})=1$  if $n=1$ and $\Delta=0$,
and $m(n, \Delta, \\ 0_=, {d}_{\leq}, {\ell}_{\leq})= m(n, \Delta, 0_{\leq}, {d}_{\leq}, {\ell}_{\leq}) =0$ if  $n \geq 2$;
\item \label{l_1_iii} $m(n, \Delta, 1_=, {d}_{=}, {\ell}_{\leq}) =0$   if $d \geq 1$, and  $m(n, \Delta, 1_=, 0_{=}, 0_{=}) = 0$ if $\Delta \geq 1$;
\item \label{l_1_iv} $m(n, \Delta, 1_=, 0_{=}, {\ell}_{=}) = m(n, \Delta, 1_=, 0_{=}, {\ell}_{\leq}) = m(n, \Delta, 1_=, j_{\leq}, {\ell}_{\leq}) = m(n, \Delta,  1_{\leq}, \\ j_{\leq}, {\ell}_{\leq}) = 1$ if $n = 2$ and $\ell = \Delta$; 
\item \label{l_1_v} $m(n, \Delta, {k}_=, {d}_=, {\ell}_=) = 0$ if  $d+\ell > \Delta$; and 
\item \label{l_1_vi} $m(n, \Delta, {k}_=, {d}_=, {\ell}_=) = m(n, \Delta, {k}_=, {d}_=, {\ell}_{\leq}) = m(n, \Delta, {k}_=, {d}_{\leq}, {\ell}_{\leq}) = 0$ if ``$k \geq n".$
\end{enumerate}
\end{lemma}
\begin{proof}

\begin{enumerate}[label=\textnormal{(\roman*)}, ref=(\roman*), font=\upshape]
\item A multigraph $M$ with max$_{\rm v}(M) = 0$ exists if and only if  ${\rm v}(M) = 1$, max$_{\rm m}(M) = 0$ and max$_{\rm l}(M) = 0$. No such tree exist with single vertex and multiple edges.
\item 
It follows from Lemma~\ref{first_recursion}(i) and (ii) that $m(n, \Delta, 0_{\leq}, {d}_{\leq}, {\ell}_{\leq}) = m(n, \Delta, 0_=,  {d}_{\leq}, {\ell}_{\leq})$.
\item Since for $d \geq 1$ ,  size of ${\rm v}(M_v)$ of the descendant subgraph must be at least $2$ and $\ell$ must be in the range as $0 < \ell \leq \Delta.$
\item  There is only one descendant subgraph of size $1$, so the only possibility for multiple edges to be on the edges $r_M{r_{M_v}}$, i.e., ${\rm e}(r_M{r_{M_v}}) = \ell$  if $\ell = \Delta.$ In case, $\ell \neq \Delta$, then there is no residual multigraph which can accommodate the remaining $\Delta-\ell$ multiple edges. 
\item The maximum number of multiple edges in the descendant subgraphs plus the tree-like multigraphs, i.e., $d+\ell$ must be less than or equal to $\Delta.$ Otherwise $d+\ell$ will exceeds $\Delta$ as both $d$ and $\ell$ have the range from $0$ to $\Delta$. 
\item The maximum size of the descendant subgraph can be $n-1.$
\end{enumerate}
\end{proof}
As a consequence of Lemma~\ref{l_1}, we have  $m(1, \Delta) = 0$ and 
$m(2, \Delta) = 1$ for $\Delta \geq 1$.
\smallskip

\subsection{Proposed Algorithm} \label{Dynamic_Programming}
After deriving the  recursive relations in Theorem~\ref{theom} and Lemma~\ref{first_recursion} and initial conditions,  we are ready to design a DP algorithm to computer 
$m(n, \Delta)$  with the recursive structures of $m(n, \Delta, {k}_{\leq}, {d}_{\leq}, {\ell}_{\leq})$, 
$m(n, \Delta, {k}_=, {d}_{\leq}, {\ell}_{\leq})$, $m(n, \Delta, {k}_=, {d}_=, {\ell}_{\leq})$, and $m(n, \Delta, {k}_=, {d}_=, {\ell}_=)$ for
$0 \leq k\leq n - 1$ and $0\leq d, \ell \leq \Delta$.
\bigskip
\begin{lemma}
\label{com_1} For any five integers $ n-1 \geq k \geq 0$, and 
$\Delta \geq d+ \ell\geq 0$, $m(n, \Delta, k_{\leq}, d_{\leq}, \ell_{\leq})$
can be obtained in 
$\mathcal{O}(nk(n + \Delta (n +  d\ell\cdot \min\{n, \Delta\})))$ time 
and 
$\mathcal{O}( nk(\Delta(d\ell+1)+1))$ space.
\end{lemma}
\begin{proof}
The proof of Lemma~\ref{com_1} follows from Algorithm~\ref{algorithm} and Lemma~\ref{lem:complexity}.
 \end{proof}

\begin{corollary}
For any two integers $n\geq 1$ and $\Delta \geq 0,$ $m(n,\Delta ,n-1_{\leq},\Delta_{\leq}, \Delta_{\leq})$
can be obtained in 
$\mathcal{O}(n^2(n + \Delta (n + \Delta^2\cdot \min\{n, \Delta\})))$ time and 
$\mathcal{O}( n^2(\Delta^3+1))$ space.
\end{corollary}
\bigskip

\noindent For any five integers $n-1\geq k\geq 0$, and $\Delta\geq d+\ell\geq 0 $, we present Algorithm~\ref{algorithm} for solving the problem of calculating 
$m(n,\Delta,k_{\leq},d_{\leq}, \ell_{\leq})$. 
In this algorithm, for each integers  
$1\leq i\leq n,0\leq j\leq \Delta ,0\leq w\leq \min\{i, k\}, 0\leq u\leq \min\{j, d\}$, and $0\leq v\leq \min\{j, \ell\}$, the variables 
$m\left[ i,j,w_{\leq},u_{\leq},v_{\leq}\right]$, $m\left[ i,j,w_=,u_{\leq},v_{\leq}\right]$, $m\left[ i,j,w_{=},u_{=},v_{\leq}\right]$, and  
$m\left[ i,j,w_=,u_=,v_=\right]$ store the values of 
$m( i,j,w_{\leq},u_{\leq},v_{\leq}), m( i,j,w_=,u_{\leq},v_{\leq}), m( i,j,w_=,u_{=},v_{\leq})$, and $m(i,j,w_=,u_=,v_=) $,~respectively. 

\bigskip
\begin{lemma}
\label{lem:complexity}
For any five integers $n-1 \geq k \geq 0$, and 
$\Delta \geq d+\ell \geq 0$, Algorithm~\ref{algorithm} outputs 
$m(n, \Delta, k_\leq, d_\leq, \ell_{\leq})$ in 
$\mathcal{O}(nk(n + \Delta (n +  d\ell\cdot \min\{n, \Delta\})))$ time 
and 
$\mathcal{O}( nk(\Delta(d\ell+1)+1))$ space.
\end{lemma}
\begin{proof}
%We give a proof for the case where $d\neq 0$, and the other case can be proved analogously.
Correctness: For each integer $1 \leq i \leq n, 0 \leq j \leq \Delta, 
0 \leq w \leq\min\{i, k\}, 
0 \leq u \leq \min\{j ,d\}$, and $0 \leq v \leq \min\{j ,\ell\}$  
all the substitutions and if-conditions in Algorithm~\ref{algorithm} follow from Theorem~\ref{theom} and Lemmas~\ref{L_1}, \ref{first_recursion}, and  \ref{l_1}.  
Furthermore, the values $m[i,j,w_\leq, u_\leq, v_\leq]$, $m[i,j,w_=, u_{\leq}, v_\leq]$,  
 $m[i,j,w_=, u_=, v_\leq]$, and $m[i,j,w_=, u_=, v_=]$ are computed by the recursive relations given in Theorem~\ref{theom} and Lemma~\ref{first_recursion}. 
This implies that Algorithm~\ref{algorithm} correctly computes the required value $m[n,\Delta, k_\leq, d_\leq, \ell_{\leq}]$.

Complexity analysis: 
There are five nested loops over the variables $i,j,w,u$ and $v$ at 
line~8, which take $\mathcal{O}(n(\Delta(kd\ell+k+1)+k +1))$ time. 
Following there are six nested loops: over variables $i,j,w,u,v$, and $p$ at lines~9, 10, 11, 12, 13 and 24, respectively.  
The loop at line~9 is of size $\mathcal{O}(n)$, 
while the loop at line~10 is of size $\mathcal{O}(\Delta)$.  
Similarly, the loops at lines~11, 12 and 13 are of size
$\mathcal{O}(k)$, $\mathcal{O}(d)$, and $\mathcal{O}(\ell)$ respectively. 
The six nested loop at line~24 is of size $\mathcal{O}(n)$
(resp., $\mathcal{O}(\min\{n, \Delta\})$ if $u = v = 0$ (resp., otherwise).  
Thus from line~9--59, %mdpi:please confirm if it correct
Algorithm~\ref{algorithm} takes $\mathcal{O}({n}^2k)$ (resp., $\mathcal{O}(nk\Delta(n + d\ell\cdot \min\{n, \Delta\}))$) time if 
$\Delta = 0$ (resp., otherwise). 
Therefore, Algorithm~\ref{algorithm} takes 
$\mathcal{O}(nk(n + \Delta (n +  d\ell\cdot \min\{n, \Delta\})))$ time. 

The algorithm stores four five-dimensional arrays. 
When $\Delta = 0$, for each integer  $1 \leq i \leq n$, and 
$1 \leq w \leq  \min\{i, k\}$ we store $m[i, 0, k_{\leq}, 0_\leq, 0_\leq], 
m[i, 0, k_=, 0_\leq, 0_\leq], m[i, 0, k_=, 0_=, 0_\leq]$ and $m[i, 0, k_=, 0_=, 0_=]$, taking $\mathcal{O}(nk)$ 
space. 
When $\Delta \geq 1$, then for each integer $1 \leq i \leq n$, 
$0 \leq j \leq \Delta$, 
$1 \leq w \leq \min\{i, k\}$,
$0 \leq u \leq \min\{j, d\}$, and $0 \leq v \leq \min\{j, \ell\}$
we store $m[i, j, w_{\leq}, u_\leq, v_\leq], 
m[i, j, w_=, u_\leq, v_\leq]$, $m[i, j, w_=, u_=, v_\leq]$, and $m[i, j, w_=, u_=, v_=]$ taking 
$\mathcal{O}(nk\Delta(d\ell+ 1))$ space. 
Hence, Algorithm~\ref{algorithm} takes $\mathcal{O}( nk(\Delta(d\ell+1)+1))$ space.
\end{proof}

\begin{algorithm} [h!]%\label{algorithm}
\caption{An algorithm for counting $%
m(n, \Delta, {k}_{\leq}, {d}_{\leq}, {\ell}_{\leq})$ based on DP}
\label{algorithm}
\begin{algorithmic}[1]
\Require Integers $n-1 \geq k \geq 0$ and $\Delta\geq d, \ell \geq 0 $.
\Ensure $m(n, \Delta, {k}_{\leq}, {d}_{\leq}, {\ell}_{\leq})$.
	\parState{$m[1, j, 0_{=},0_=, 0_=] := m[1, j, 0_=, 0_=, 0_{\leq}] := m[1, j, 0_=, 0_{\leq}, 0_{\leq}] :=\\
			    m[1, j,0_{\leq}, 0_{\leq}, 0_{\leq}] := 1$;}
		\parState{$m[2, j, 1_=, 0_=, v_=] := m[2, j, 1_=, 0_=, j_{\leq}] := 
			     m[2, j, 1_=, j_{\leq}, j_{\leq}] :=\\
			     m[2, j,  1_{\leq}, j_{\leq}, j_{\leq}] := 1$ if $v = j$;}
                \parState{$m[i, j, w_=, u_{\leq}, v_{\leq}] := 0$ if $w$ $\geq$ $i$;}
                \parState{$m[1, j, 0_{\leq}, u_{\leq}, v_{\leq}] := 0$ if $j$ $\geq$ $1$;}
                \parState{$m[i, j, 0_=, u_{\leq}, v_{\leq}] := m[i, j, 0_{\leq}, u_{\leq}, v_{\leq}] :=0$ if $i$ $\geq$ $2$;}
			  \parState{$m[i, j, 1_=, u_=, v_{\leq}] := 0$ if $u$ $\geq$ $1$;}
		      \parState{$m[i, j, w_=, u_=, v_=] := 0;$ if $u+v$ $>$ $j$;}
            \parState{for each $ 2 \leq i \leq n,  0\leq j \leq \Delta, 1 \leq w \leq \min\{i, k\}, 0\leq u \leq 
                \min\{j, d\},  0\leq v \leq \min\{j, \ell\}$.}  
   \For{$i :=  3, \ldots, n$}
	\For{$j := 0, \ldots, \Delta$}	
		\For{$w := 1,  \ldots, \min\{i, k\}$} 
			\For{$u := 0,  \ldots, \min\{j,d\}$}  
                \For{$v := 0,  \ldots, \min\{j,\ell\}$}
                  \If{$ u + v \leq j$}  
				    \If{$j = u = v = 0$ and $w = 1$} 
					  \State{$m[i, 0, 1_=, 0_=, 0_=]:= m[i, 0, 1_=, 0_=, 0_{\leq}]:= m[i, 0, 1_=, 0_{\leq},$ 
					 \Statex \hspace{3.7cm} $0_{\leq}]:= m[i, 0, 1_{\leq},0_{\leq}, 0_{\leq}]:=1$}
				         \Else{ $w \geq 2$ or $j \geq 1$ or $u+v \geq 1$} 
					      \State{$h := 1$; $m[i,j,w_=,u_=,v_=] :=0;$ }  
					       \If{$u = v = 0$}   
						      \State{$z := \lfloor (i - 1)/w \rfloor$} 
					           \Else{ $u \geq 1 $ or $v \geq 1$}
					           \State{$z := \min \{\lfloor (i-1)/w \rfloor, \lfloor j/u+v \rfloor \}$} 
                                       \EndIf
					         \For{$p := 1, 2, \ldots, z$}  	
						      \State{%$k ^* := \min\{i - kq - 1, k - 1\}$; 
						       $h := h  \cdot 
						          (m[w, u, w-1_\leq, u_\leq, u_\leq] + p -1)/p;$}
						           %\State{}
						           \If{$u = v = 0$} 
							           \State{$m[i,j,w_=,u_=,v_=] := m[i,j,w_=,u_=,v_=] + h \cdot m[i-$  
							           \Statex \hspace{4.6cm} $pw,j,\min\{i - pw - 1, w - 1\}_{\leq},j_{\leq},j_{\leq}]$}
                                        \ElsIf{$u\geq 1$ and $v = 0$ } 
				                        \State{$m[i,j,w_=,u_=,v_=] := m[i,j,w_=,u_=,v_=] + h \cdot m[i-$  
						              \Statex \hspace{4.6cm} $pw,j-pu,w_=,\min\{j - pu, u - 1\}_{\leq}, j-pu_{\leq}]+  m[i$
						             \Statex \hspace{4.6cm} $-pw,j-pu,\min\{i - pw - 1, w-1\}_{\leq},j-pu_{\leq}, j-pu_{\leq}]$}
                                         \ElsIf{ $u = 0$ and $v\geq 1$}
                                           \State{$m[i,j,w_=,u_=,v_=] := m[i,j,w_=,u_=,v_=] +  h \cdot m[i-$
                                            \Statex \hspace{4.8cm}$pw,j-pv,w_=, u_=, \min\{j - pv, v - 1\}_{\leq}]+ m[i-pw,$
                                           \Statex \hspace{4.7cm} $j-pv,\min\{i - pw - 1, w - 1\}_{\leq},j-pv_{\leq}, j-pv_{\leq}]$}
                                          \ElsIf{ $u\geq 1$ and $v\geq 1$} 
  \algstore{bkbreak}
\end{algorithmic}
\end{algorithm}
\addtocounter{algorithm}{-1}
\begin{algorithm}[h]
\caption{Continued} 
\begin{algorithmic}[1]
\algrestore{bkbreak}	 
  \State{$m[i,j,w_=,u_=,v_=] := m[i,j,w_=,u_=,v_=] +  h\cdot m[i-$
  \Statex \hspace{4.7cm}$pw,j-p(u+v),w_=, u_=, \min\{j - p(u+v), v - 1\}_{\leq}]+  $
  \Statex \hspace{4.7cm}$m[i-pw,j-p(u+v), w_=, \min\{j - p(u+v), u -1\}_{\leq},$
   \Statex \hspace{4.7cm}$  j-p(u+v)_{\leq}]+ m[i-pw,j-p(u+v), \min\{i - pw - 1, $  
   \Statex \hspace{4.7cm}$ w - 1\}_{\leq},j-p(u+v)_{\leq}, j-p(u+v)_{\leq}]$}   
                                \EndIf 
                                \EndFor 
\If{ $u = v = 0$} 
						       \State{$m[i,j,w_=,0_{\leq},0_{\leq}] : = m[i,j,w_=,0_=,0_=]$}  
					          \ElsIf{ $u = 0$ and $v \geq 1$} 
						       \State{$m[i,j,w_=,0_{\leq},v_{\leq}] := m[i,j,w_=,0_{\leq},{v-1}_{\leq}]+ m[i,j,$
							   \Statex \hspace{4.1cm} $w_=, 0_=,v_=]$};
							   \State{$m[i,j,w_=,u_=,v_{\leq}] := m[i,j,w_=,u_=,{v-1}_{\leq}] + m[i,j,$
							   \Statex \hspace{4.1cm} $w_=,u_=,v_=]$}; 
                                \ElsIf{ $u \geq 1$ and $v = 0$} 
						      \State{$m[i,j,w_=,u_{\leq},0_{\leq}] := m[i,j,w_=,{u-1}_{\leq},0_{\leq}] + m[i,j,$
						      \Statex \hspace{4.2cm}$w_=,u_=,0_=]$}; 
							   \State{$m[i, j, w_=, u_=, v_{\leq}] := m[i, j, w_=, u_=, 0_=]$} 
                                \Else{ $u \geq 1$ and $v \geq 1$} 
                                 \State{$m[i,j,w_=,u_=,v_{\leq}] := m[i,j,w_=,u_=,{v-1}_{\leq}] + m[i,j,$
                                 \Statex \hspace{4.2cm}$w_=,u_=,v_=]$};
                                 \State{$m[i,j,w_=,u_{\leq},v_{\leq}] := m[i,j,w_=,{u-1}_{\leq},v_{\leq}] + m[i,j,$
                                 \Statex \hspace{4.2cm}$w_=,u_=,v_{\leq}]$} 
								\EndIf 
					           \State{$m[i,j,w_{\leq},u_{\leq},v_{\leq}] := m[i,j,{w-1}_{\leq},u_{\leq},v_{\leq}] + m[i,j,w_=,$
						       \Statex \hspace{3.7cm}$u_{\leq},v_{\leq}]$}
                     \EndIf 
                   \Else % u+v>j
                    \State{$m[i,j,w_=,u_=,v_{\leq}] := m[i,j,w_=,u_=,{v-1}_{\leq}] + m[i,j,w_=,u_=,$
                    \Statex \hspace{3.2cm}$v_=]$} ;
                                 \State{$m[i,j,w_=,u_{\leq},v_{\leq}] := m[i,j,w_=,{u-1}_{\leq},v_{\leq}] + m[i,j,w_=,u_=,$
                                \Statex \hspace{3.1cm} $v_{\leq}]$};
                                 \State{$m[i,j,w_{\leq},u_{\leq},v_{\leq}] := m[i,j,{w-1}_{\leq},u_{\leq},v_{\leq}] + m[i,j,w_=,u_{\leq},$
                                \Statex \hspace{3.1cm} $v_{\leq}]$}
				    \EndIf 
				\EndFor
			\EndFor 
		\EndFor 
	\EndFor 
\EndFor 	
\State{{\bf Return} $m[n,\Delta,{k}_{\leq},d_{\leq},\ell_{\leq}]$  as $m(n, \Delta,{k}_{\leq},d_{\leq},\ell_{\leq})$.}
\end{algorithmic}
\end{algorithm}
\restoregeometry

%Note that any tree can be considered as rooted tree in a unique way by considering unicentroid as the root or, in a bicentroid case, by adding a virtual vertex on the bicentroid and taking that virtual vertex to be the root.
\begin{theorem}
For any two integers $n \geq 1$ and $\Delta \geq 0$, the number of non-isomorphic tree-like multigraphs with $n$ vertices and 
$\Delta$ multiple edges can be obtained in 
$\mathcal{O}(n^2(n + \Delta (n + \Delta^2\cdot \min\{n, \Delta\})))$ time and 
$\mathcal{O}( n^2(\Delta^3+1))$ space.
\end{theorem}
\begin{proof}
Each multigraph becomes a tree after removing multiple edges. Consequently, we can classify the multigraphs of $\mathcal{M}(n,\Delta)$ based on centroids of the  trees  by applying Jordan result~\cite{jordan1869assemblages}. This classification divides multigraphs into two categories: those with a bicentroid and those with a unicentroid. The bicentroid case holds when the number of vertices is  even, resulting in a multigraph with exactly two descendent subgraph of size $n/2$. In the case of unicentroid, the size of  descendent subgraph cannot exceed $\lfloor {n}/2 \rfloor$. \comb{An illustration of unicentroid is given in Fig.~\ref{fig:centroid_bi-uni}(a)}.
\begin{figure}[h!]
    \centering
    \includegraphics[width=0.9\textwidth]{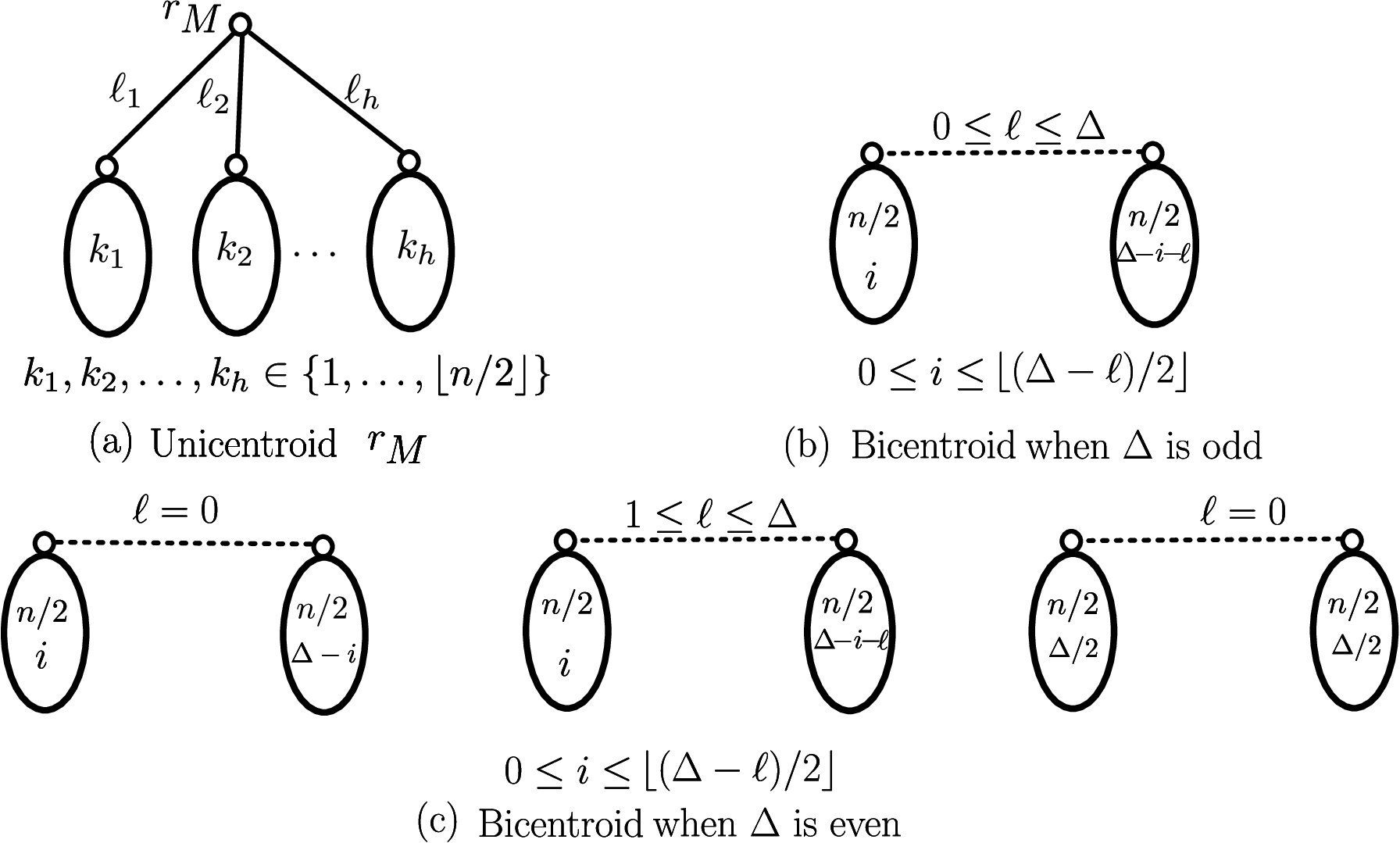}
    \caption{\comb{Illustration of multigraphs with unicentroid and bicentroid.}}
    \label{fig:centroid_bi-uni}
\end{figure}
  
  Let $n$ be an even integer. 
  Then any tree $M$ with $n$ vertices, 
  $\Delta$ multiple edges and a bicentroid 
    has two connected components, $X$ and $Y$ obtained  by the 
  removal of the bicentroid such that 
 $ X \in \mathcal{M}(n/2, i, n/2 - 1_{\leq}, i_{\leq}, i_{\leq})$ and 
 $Y \in \mathcal{M}(n/2, \Delta - i - \ell, n/2 - 1_{\leq},\Delta - i - \ell_{\leq},\Delta - i - \ell_{\leq})$
for some $0 \leq i \leq \lfloor \Delta/2 \rfloor$, 
where if $\Delta $ is even and $\ell = 0$ then for $i = \Delta /2$, both of the 
components $X$ and $Y$ belong to 
$\mathcal{M}(n/2, \Delta/2, n/2 - 1_{\leq}, \Delta/2_{\leq}, \Delta/2_{\leq})$. 
Note that for any $0 \leq i \leq \lfloor (\Delta-\ell)/2 \rfloor$, it holds that 
\[\mathcal{M}(n/2, i, n/2 - 1_{\leq}, i_{\leq}, i_{\leq}) \cap 
\mathcal{M}(n/2, \Delta - i - \ell, n/2 - 1_{\leq},\Delta - i - \ell_{\leq}, \Delta - i - \ell_{\leq}) = \emptyset.\] 
Therefore, when $\Delta$ is odd,  
the number of mutually non-isomorphic trees with $n$ vertices, 
$\Delta$ multiple edges, and 
 a bicentroid is \\

$\sum\limits_{i = 0}^{\mathclap{\lfloor (\Delta-\ell) /2 \rfloor }}\ \ \ \sum\limits_{\ell = 0}^{\mathclap{\Delta }}
m(n/2, i, n/2 - 1_{\leq}, i_{\leq}, i_{\leq}) $
 $m(n/2, \Delta - i - \ell, n/2 - 1_{\leq},\Delta - i - \ell_{\leq}, \Delta - i - \ell_{\leq})$. 
\comb{For odd $\Delta$, the representation of the bicentroid is given in Fig.~\ref{fig:centroid_bi-uni}(b).}

When $\Delta$ is even, the number of mutually non-isomorphic trees with $n$ vertices, $\Delta$ multiple edges, and a bicentroid is \\

$\sum\limits_{i = 0}^{\mathclap{\lfloor (\Delta-1) /2 \rfloor }}
m(n/2, i, n/2 - 1_{\leq}, i_{\leq}, i_{\leq}) 
 m(n/2, \Delta - i, n/2 - 1_{\leq},\Delta - i_{\leq}, \Delta - i_{\leq}) +$ \\
 
$ \sum\limits_{i = 0}^{\mathclap{\lfloor (\Delta-\ell) /2 \rfloor }} \ \ \  \sum\limits_{\ell = 1}^{\mathclap{\Delta }} m(n/2, i, n/2 - 1_{\leq}, i_{\leq}, i_{\leq}) 
 m(n/2, \Delta - i - \ell, n/2 - 1_{\leq},\Delta - i - \ell_{\leq}, \Delta - i - \ell_{\leq}) +$ 
 
 \begin{equation*}
     \binom{m(n/2, \Delta/2, n/2 - 1_{\leq}, \Delta/2_\leq, \Delta/2_\leq) + 1}{2}. \\ 
 \end{equation*}
          
 \comb{For even $\Delta$, the representation of the bicentroid is given in Fig.~\ref{fig:centroid_bi-uni}(c).}          
 
Thus, when $\Delta$ is odd, the number of mutually non-isomorphic trees with $n$ vertices and
$\Delta$ multiple edges is

%\begin{equation}
%\label{total_tree}
%\begin{split}
 $m(n, \Delta, \lfloor (n-1)/2 \rfloor_{\leq}, \Delta_{\leq}, \Delta_{\leq})$+ \\

$\sum\limits_{i = 0}^{\mathclap{\lfloor (\Delta-\ell) /2 \rfloor }} \ \ \ \sum\limits_{\ell = 0}^{\mathclap{\Delta }}
m(n/2, i, n/2 - 1_{\leq}, i_{\leq}, i_{\leq}) 
 m(n/2, \Delta - i - \ell, n/2 - 1_{\leq},\Delta - i - \ell_{\leq}, \Delta - i - \ell_{\leq}).$\\
%\end{split}
%\end{equation}
 
 When $\Delta$ is even, the number of mutually non-isomorphic trees with $n$ vertices and
$\Delta$ multiple edges is \\
%\begin{equation}
%\label{total_tree}
%\begin{split}
$m(n, \Delta, \lfloor (n-1)/2 \rfloor_{\leq}, \Delta_{\leq}, \Delta_{\leq})$ + \\

$ \sum\limits_{i = 0}^{\mathclap{\lfloor (\Delta-1) /2 \rfloor }} 
m(n/2, i, n/2 - 1_{\leq}, i_{\leq}, i_{\leq})\ 
m(n/2, \Delta - i, n/2 - 1_{\leq}, \Delta - i_{\leq}, \Delta - i_{\leq}) +$ \\ 

$\sum\limits_{i = 0}^{\mathclap{\lfloor (\Delta-\ell) /2 \rfloor }} \ \ \ \sum\limits_{\ell = 1}^{\mathclap{\Delta }} m(n/2, i, n/2 - 1_{\leq}, i_{\leq}, i_{\leq})\ 
m(n/2, \Delta - i - \ell, n/2 - 1_{\leq}, \Delta - i -\ell_{\leq}, \Delta - i - \ell_{\leq})+$
 
 \begin{equation*}
 \binom{m(n/2, \Delta/2, n/2 - 1_{\leq}, \Delta/2_\leq, \Delta/2_\leq) + 1}{2}.
 %\end{split}
\end{equation*}

\end{proof}
\subsection{Experimental results}\label{exp}
We implemented the proposed algorithm using Python 3.10 and tested it for counting tree-like multigraphs with a given number of vertices and multiple edges. The tests were conducted on a machine with Processor: 12th Gen, Core i7(1.7 GHz) and Memory: 16 GB RAM. 
The experimental results presented in Table~\ref{tab_time} demonstrate that the proposed method efficiently 
counted tree-like multigraphs with up to  $170$ vertices and $50$ multiple edges in atmost $930$ seconds. 

\begin{table}[h!]
\caption{Experimental result of the counting method. }
\centering
\label{tab_time}
\begin{tabular}{>{$}c<{$} | S[table-format=2.2e1] | S}
    \toprule
    \boldsymbol{(n, \Delta)} & \textbf{Number of Tree-Like Multigraphs} & \textbf{Time [sec.]} \\
    \midrule
$(8, 5)$ & \num{4.41e3} & 0.0115 \\
$(15, 10)$ & \num{2.93e+9} &0.2371\\
$(18, 13)$ & \num{2.40e+12} &0.5885 \\
$(22, 15)$ & \num{4.41e+15} &1.6876 \\
$(25, 18)$ & \num{4.21e+18} & 3.2512\\
$(30, 20)$ & \num{4.05e+22} &5.9189 \\
$(35, 15)$ & \num{4.43e+23} & 3.4999\\
$(40, 20)$ & \num{8.54e+28} & 10.6742\\
$(43, 33)$ & \num{6.61e+35} &49.2939\\
$(45, 36)$ & \num{1.56e+38} &69.7175 \\
$(50, 44)$ & \num{1.77e+44} &175.0787 \\
$(66, 48)$ & \num{5.17e+56} &407.9641\\
$(72, 50)$ & \num{3.16e+61} &549.2499\\
$(80, 40)$ & \num{4.52e+62} &342.5029\\
$(84, 45)$ & \num{2.35e+67} &549.1110\\
$(90, 50)$ & \num{2.14e+73} &821.8740\\
$(95, 50)$ & \num{3.32e+76} &932.2500\\
$(100, 40)$ & \num{7.83e+74} &557.7977\\
$(120, 35)$ & \num{5.95e+83} &522.7298\\
$(150, 30)$ & \num{1.11e+97} &545.3231\\
$(170, 25)$ & \num{5.92e+103} &426.8230\\
\bottomrule
\end{tabular}
\end{table}
\section{Conclusion}\label{conclusion}
We proposed an algorithm based on dynamic programming to count the number of mutually distinct tree-like multigraphs with $n$ vertices and $\Delta$ multiple edges for $n \geq 1$ and $\Delta \geq 0.$ To achieve this, we define a specific ordering to avoid isomorphism. We then formulated recursive relations and identified corresponding initial conditions. Using these recursive relations, we developed a DP algorithm. The algorithm was implemented and tested on several instances, and results clearly demonstrate the efficiency of the proposed method. A natural direction for future research is to count the number of mutually distinct tree-like multigraphs with a prescribed number of vertices, multiple edges and self-loops. Additionally, it would be interesting to design an efficient algorithm for generating tree-like mulitgraphs by using the counting results established in this study.

\bibliography{ref}

\end{document}